\documentclass[reqno]{amsart}
\usepackage{amsmath,amsthm,amssymb,eucal,mathrsfs,fullpage,setspace,bm,mathabx}
\usepackage{graphicx}
\usepackage[round]{natbib}
\usepackage[all,cmtip]{xy}
\usepackage{lineno}
\usepackage{subfig}

\title{autocratic strategies for iterated games with arbitrary action spaces}
\author{alex mcavoy and christoph hauert}

\theoremstyle{definition}
\newtheorem{corollary}{Corollary}
\newtheorem{definition}{Definition}

\newtheorem{lemma}{Lemma}
\newtheorem{proposition}{Proposition}
\newtheorem{remark}{Remark}
\newtheorem{theorem}{Theorem}

\newcommand{\T}{\intercal}
\newcommand{\eq}[1]{Eq.~(\ref{eq:#1})}
\newcommand{\fig}[1]{Fig.~\ref{fig:#1}}
\newcommand{\thm}[1]{Theorem~\ref{thm:#1}}

\begin{document}

\begin{abstract}
The recent discovery of zero-determinant strategies for the iterated Prisoner's Dilemma sparked a surge of interest in the surprising fact that a player can exert unilateral control over iterated interactions. These remarkable strategies, however, are known to exist only in games in which players choose between two alternative actions such as ``cooperate" and ``defect." Here we introduce a broader class of autocratic strategies by extending zero-determinant strategies to iterated games with more general action spaces. We use the continuous Donation Game as an example, which represents an instance of the Prisoner's Dilemma that intuitively extends to a continuous range of cooperation levels. Surprisingly, despite the fact that the opponent has infinitely many donation levels from which to choose, a player can devise an autocratic strategy to enforce a linear relationship between his or her payoff and that of the opponent even when restricting his or her actions to merely two discrete levels of cooperation. In particular, a player can use such a strategy to extort an unfair share of the payoffs from the opponent. Therefore, although the action space of the continuous Donation Game dwarfs that of the classical Prisoner's Dilemma, players can still devise relatively simple autocratic and, in particular, extortionate strategies.
\end{abstract}

\maketitle

\section{Introduction}

Game theory provides a powerful framework to study interactions between individuals (``players"). Among the most interesting types of interactions are social dilemmas, which result from conflicts of interest between individuals and groups \citep{dawes:ARP:1980,vanlange:OBHDP:2013}. Perhaps the most well-studied model of a social dilemma is the Prisoner's Dilemma \citep{axelrod:Science:1981}. A two-player game with actions, $C$ (``cooperate") and $D$ (``defect"), and payoff matrix,
\begin{align}\label{eq:classicalPD}
\bordermatrix{%
 & C & D \cr
C &\ R & \ S \cr
D &\ T & \ P \cr
} ,
\end{align}
is said to be a Prisoner's Dilemma if $T>R>P>S$ \citep{axelrod:BB:1984}. In a Prisoner's Dilemma, defection is the dominant action, yet the players can realize higher payoffs from mutual cooperation ($R$) than they can from mutual defection ($P$), resulting in a conflict of interest between the individual and the pair, which characterizes social dilemmas. Thus, in a one-shot game (i.e. a single encounter), two opponents have an incentive to defect against one another, but the outcome of mutual defection (the unique Nash equilibrium) is suboptimal for both players.

One proposed mechanism for the emergence of cooperation in games such as the Prisoner's Dilemma is direct reciprocity \citep{trivers:TQRB:1971,nowak:Science:2006}, which entails repeated encounters between players and allows for reciprocation of cooperative behaviors. In an iterated game, a player might forgo the temptation to defect in the present due to the threat of future retaliation--``the shadow of the future"--or the possibility of future rewards for cooperating \citep{axelrod:BB:1984,delton:PNAS:2011}, phenomena for which there is both theoretical and empirical support \citep{heide:AMJ:1992,bo:AER:2005}. One example of a strategy for the iterated game is to copy the action of the opponent in the previous round (``tit-for-tat") \citep{axelrod:BB:1984}. Alternatively, a player might choose to retain his or her action from the previous round if and only if the most recent payoff was $R$ or $T$ (``win-stay, lose-shift") \citep{nowak:Nature:1993}. These examples are among the simplest and most successful strategies for the iterated Prisoner's Dilemma \citep{nowak:BP:2006}.

In a landmark paper, \citet{press:PNAS:2012} deduce the existence of zero-determinant strategies, which allow a single player to exert much more control over this game than previously thought possible. Since their introduction, these strategies have been extended to cover multiplayer social dilemmas \citep{hilbe:PNAS:2014,pan:SR:2015} and temporally-discounted games \citep{hilbe:GEB:2015}. Moreover, zero-determinant strategies have been studied in the context of evolutionary game theory \citep{hilbe:PNAS:2013,adami:NC:2013,stewart:PNAS:2013,szolnoki:PRE:2014a,hilbe:JTB:2015}, adaptive dynamics \citep{hilbe:PLOSONE:2013}, and human behavioral experiments \citep{hilbe:NC:2014}. In each of these studies, the game is assumed to have only two actions: cooperate and defect. In fact, the qualifier ``zero-determinant" actually reflects this assumption because these strategies force a matrix determinant to vanish for action spaces with only two options. We show here that this assumption is unnecessary.

More specifically, suppose that players $X$ and $Y$ interact repeatedly with no limit on the number of interactions. For games with two actions, $C$ and $D$, a memory-one strategy for player $X$ is a vector, $\mathbf{p}=\left(p_{CC},p_{CD},p_{DC},p_{DD}\right)^{\T}$, where $p_{xy}$ is the probability that $X$ cooperates following an outcome in which $X$ plays $x$ and $Y$ plays $y$. Let $\mathbf{s}_{X}=\left(R,S,T,P\right)^{\T}$ and $\mathbf{s}_{Y}=\left(R,T,S,P\right)^{\T}$ be the payoff vectors for players $X$ and $Y$, respectively, and let $\alpha$, $\beta$, and $\gamma$ be fixed constants. \citet{press:PNAS:2012} show that if there is a constant, $\phi$, for which
\begin{align}\label{eq:pressDysonVector}
\widetilde{\mathbf{p}} &:= \begin{pmatrix}p_{CC}-1 \\ p_{CD}-1 \\ p_{DC} \\ p_{DD}\end{pmatrix} = \phi \left( \alpha\mathbf{s}_{X} + \beta\mathbf{s}_{Y} + \gamma \right) ,
\end{align}
then $X$ can unilaterally enforce the linear relationship $\alpha\pi_{X}+\beta\pi_{Y}+\gamma =0$ on the average payoffs, $\pi_{X}$ and $\pi_{Y}$, by playing $\mathbf{p}$. A strategy, $\mathbf{p}$, that satisfies \eq{pressDysonVector} is known as a ``zero-determinant" strategy due to the fact that $\widetilde{\mathbf{p}}$ causes a particular matrix determinant to vanish \citep{press:PNAS:2012}. However, what is important about these strategies is not that they cause some matrix determinant to vanish, but rather that they unilaterally enforce a linear relationship on expected payoffs. Therefore, we refer to these strategies and their generalization to arbitrary action spaces as autocratic strategies. Of particular interest are extortionate strategies, which ensure that a player receives an unfair share of the payoffs exceeding the payoff at the Nash equilibrium \citep{stewart:PNAS:2012}. Hence, if $P$ is the payoff for mutual defection in the Prisoner's Dilemma, then $\mathbf{p}$ is a an extortionate strategy for player $X$ if $\mathbf{p}$ enforces the equation $\pi_{X}-P=\chi\left(\pi_{Y}-P\right)$ for some extortion factor, $\chi\geqslant 1$.

The most common extensions of finite action sets are continuous action spaces. An element $s\in\left[0,K\right]$ represents a player's investment or cooperation level (up to some maximum, $K$), such as the amount a player invests in a public good \citep{doebeli:EL:2005}; the volume of blood one vampire bat donates to another \citep{wilkinson:Nature:1984}; the amount of resources used by microbes to produce siderophores \citep{west:PRSB:2003}; or the effort expended in intraspecies grooming \citep{hemelrijk:AB:1994,akinyi:AB:2013}. It is important to note that games with continuous action spaces can yield qualitatively different results than their discrete counterparts. 
For example, the strategy ``raise-the-stakes" initially offers a small investment in Prisoner's Dilemma interactions and subsequently raises the contribution in discrete increments if the opponent matches or betters the investment \citep{roberts:Nature:1998}. However, in a continuous action space, raise-the-stakes evolves into defection due to the fact that another strategy can be arbitrarily close--in terms of the initial investment and subsequent increases in contribution--yet exhibit qualitatively different behavior \citep{killingback:Nature:1999}. In particular, raise-the-stakes succeeds in a discrete action space but fails in a continuous one.

\citet{akin:Games:2015} calls the vector, $\widetilde{\mathbf{p}}$, of \eq{pressDysonVector} a Press-Dyson vector. For continuous action spaces, the payoff vectors, $\mathbf{s}_{X}$ and $\mathbf{s}_{Y}$, must be replaced by payoff functions, $u_{X}\left(x,y\right)$ and $u_{Y}\left(x,y\right)$. That is, $u_{i}\left(x,y\right)$ denotes the payoff to player $i$ when $X$ plays $x$ and $Y$ plays $y$. The analogue of the linear combination $\alpha\mathbf{s}_{X}+\beta\mathbf{s}_{Y}+\gamma$ is the function $\alpha u_{X} + \beta u_{Y} + \gamma$. Here, we formally define a Press-Dyson function that extends the Press-Dyson vector to iterated games with arbitrary action spaces. This extension allows one to deduce the existence of strategies that unilaterally enforce linear relationships on the payoffs in more general iterated games. In particular, autocratic (or zero-determinant strategies) are not peculiar to games with two (or even finitely many) actions. Moreover, under mild conditions, player $X$ can enforce a linear relationship on expected payoffs by choosing a memory-one strategy that plays just two actions, despite the fact that the opponent may have an infinite number of actions from which to choose. We give examples of such autocratic strategies in the continuous Donation Game, which represents an instance of the Prisoner's Dilemma but with an extended, continuous action space.

\section{Autocratic strategies}
Consider a two-player iterated game with actions spaces, $S_{X}$ and $S_{Y}$, and payoff functions, $u_{X}\left(x,y\right)$ and $u_{Y}\left(x,y\right)$, for players $X$ and $Y$, respectively. Players $X$ and $Y$ interact repeatedly (infinitely many times), deriving a payoff at each round based on $u_{X}$ and $u_{Y}$. We treat games with temporally-discounted payoffs, which means that for some discounting factor $\lambda$ with $0<\lambda <1$, a payoff of $1$ at time $t+\tau$ is treated the same as a payoff of $\lambda^{\tau}$ at time $t$ \citep[see][]{fudenberg:MIT:1991}. Alternatively, one may interpret this game as having a finite number of rounds, where in any given round $\lambda$ denotes the probability of another round \citep{nowak:Science:2006}, which results in an expected game length of $1/\left(1-\lambda\right)$ rounds.

Suppose that, for each $t\geqslant 0$, $x_{t}$ and $y_{t}$ are the actions used by players $X$ and $Y$ at time $t$. Then, irrespective of the interpretation of $\lambda$, the average payoff to player $X$ is
\begin{align}\label{eq:totalDiscountedPayoff}
\pi_{X} &= \left(1-\lambda\right)\sum_{t=0}^{\infty}\lambda^{t}u_{X}\left(x_{t},y_{t}\right) .
\end{align}
The payoff to player $Y$, $\pi_{Y}$, is obtained by replacing $u_{X}$ with $u_{Y}$ in Eq. (\ref{eq:totalDiscountedPayoff}). If the strategies of $X$ and $Y$ are stochastic, then the payoffs are random variables with expectations, $\pi_{X}$ and $\pi_{Y}$ (see \textbf{Supporting Information}). Of particular interest are memory-one strategies, which are probabilistic strategies that depend on only the most recent outcome of the game. If $\sigma_{X}$ is a memory-one strategy for $X$, then we denote by $\sigma_{X}\left[x,y\right]\left(s\right)$ the probability that $X$ uses $s$ after $X$ plays $x$ and $Y$ plays $y$ (see Fig. \ref{fig:density}).

\begin{figure*}
\begin{center}
\includegraphics[scale=0.5]{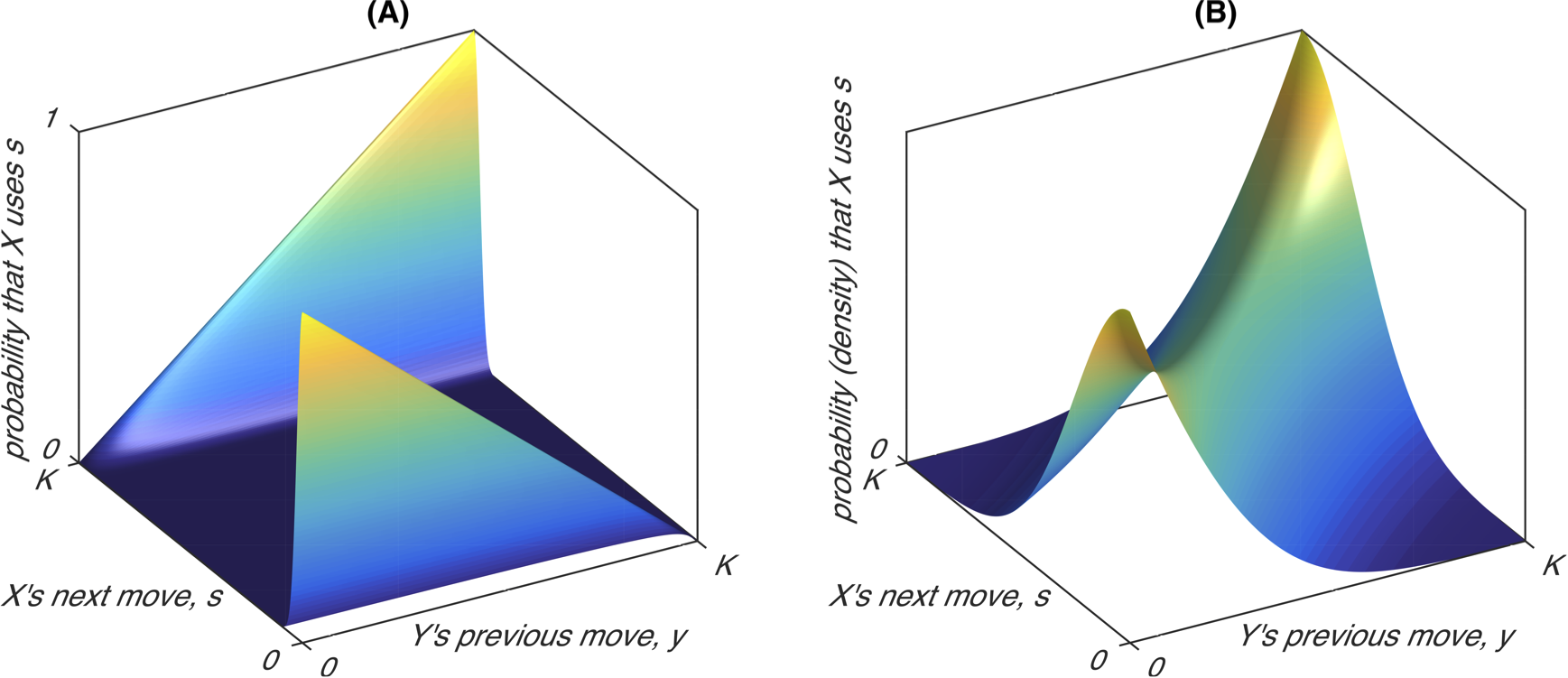}
\end{center}
\caption{Reactive, memory-one strategies, $\sigma_{X}\left[x,y\right]\left(s\right) =\sigma_{X}\left[y\right]\left(s\right)$, for a game whose action space is the interval $\left[0,K\right]$. In (A), player $X$ uses $Y$'s action in the previous round to determine the probabilities with which she plays $0$ and $K$ in the next round. As $Y$'s previous action (investment level), $y$, increases, so does the probability that $X$ uses $K$ in the subsequent round. Since $X$ plays only these two actions, this strategy is called a two-point strategy. In (B), player $X$ uses $Y$'s action in the previous round to determine the probability density function she uses to devise her next action. In contrast to (A), which depicts a strategy concentrated on just two actions, the strategy depicted in (B) is concentrated on a continuous range of actions in $\left[0,K\right]$. As $Y$'s previous action, $y$, increases, the mean of the density function governing $X$'s next action increases, which indicates that $X$ is more willing reciprocate by increasing her investment (action) in response to $Y$ increasing his investment.
\label{fig:density}}
\end{figure*}

The proofs of the existence of zero-determinant strategies (both in games with and without discounting) rely heavily on the fact that the action space is finite \citep{press:PNAS:2012,hilbe:GEB:2015,akin:Games:2015}. In particular, it remains unclear whether zero-determinant strategies are consequences of the finiteness of the action space or instances of a more general concept. Here, we introduce autocratic strategies as an extension of zero-determinant strategies to discounted games with arbitrary action spaces. The traditional, undiscounted case is recovered in the limit $\lambda\rightarrow 1$.
\begin{theorem}[Autocratic strategies]\label{thm:mainTheorem}
Suppose that $\sigma_{X}\left[x,y\right]$ is a memory-one strategy for player $X$ and let $\sigma_{X}^{0}$ be player $X$'s initial action. If, for some bounded function, $\psi$, the equation
\begin{align}\label{eq:mainEquation}
\alpha u_{X}\left(x,y\right) + \beta u_{Y}\left(x,y\right) + \gamma &= \psi\left(x\right) - \lambda\int_{s\in S_{X}}\psi\left(s\right)\,d\sigma_{X}\left[x,y\right]\left(s\right)  - \left(1-\lambda\right)\int_{s\in S_{X}}\psi\left(s\right)\,d\sigma_{X}^{0}\left(s\right)
\end{align}
holds for each $x\in S_{X}$ and $y\in S_{Y}$, then $\sigma_{X}^{0}$ and $\sigma_{X}\left[x,y\right]$ together enforce the linear payoff relationship
\begin{align}\label{eq:linearRelationship}
\alpha\pi_{X}+\beta\pi_{Y}+\gamma &= 0
\end{align}
for \emph{any} strategy of player $Y$. In other words, the pair $\left(\sigma_{X}^{0},\sigma_{X}\left[x,y\right]\right)$ is an autocratic strategy for player $X$.
\end{theorem}

Note that the initial action, $\sigma_{X}^{0}$, in Eq. (\ref{eq:mainEquation}) becomes irrelevant without discounting ($\lambda\to1$). The function $\psi$ may be interpreted as a ``scaling function" that is used to ensure $\sigma_{X}\left[x,y\right]$ is a feasible memory-one strategy; that is, $\psi$ plays the same role as the scalar $\phi$ in \eq{pressDysonVector}, which is chosen so that the entries of $\mathbf{p}$ are all between $0$ and $1$. We call the right-hand side of Eq. (\ref{eq:mainEquation}) a Press-Dyson function, which extends the Press-Dyson vector of \eq{pressDysonVector} to arbitrary action spaces (see \textbf{Supporting Information}). In contrast to action spaces with two options (``cooperate" and ``defect", for instance), autocratic strategies are defined only implicitly via Eq. (\ref{eq:mainEquation}) for general action spaces (and actually already for games with just three actions).

For each $x$ and $y$, the integral, $\int_{s\in S_{X}}\psi\left(s\right)\,d\sigma_{X}\left[x,y\right]\left(s\right)$, may be thought of as the weighted average (expectation) of $\psi$ with respect to $\sigma_{X}\left[x,y\right]$. Since the integral is taken against $\sigma_{X}\left[x,y\right]$, in general one cannot solve \eq{mainEquation} explicitly for $\sigma_{X}\left[x,y\right]$, so it is typically not possible to directly specify all pairs $\left(\sigma_{X}^{0},\sigma_{X}\left[x,y\right]\right)$ that unilaterally enforce \eq{linearRelationship}.

Interestingly, under mild conditions, $\sigma_{X}$ can be chosen to be a remarkably simple ``two-point" strategy, concentrated on just two actions, $s_{1}$ and $s_{2}$ (see Corollary \ref{cor:mainCorollary} in \textbf{Supporting Information}). Player $X$ can enforce \eq{linearRelationship} by playing either $s_{1}$ or $s_{2}$ in each round, with probabilities determined by the outcome of the previous round (see \fig{density}(A)). Thus, a strategy of this form uses the (memory-one) history of previous play only to adjust the relative weights placed on $s_{1}$ and $s_{2}$, while $s_{1}$ and $s_{2}$ themselves remain unchanged. Unlike in the case of arbitrary $\sigma_{X}$, for fixed $\psi$, $s_{1}$, and $s_{2}$, it is possible to explicitly solve for all autocratic, two-point strategies on $s_{1}$ and $s_{2}$ satisfying \eq{mainEquation} (see Remark \ref{rem:twoPointExplicit} in \textbf{Supporting Information}). In a two-action game, every memory-one strategy is concentrated on two points, which explains why games like the classical Prisoner's Dilemma fail to capture the implicit nature of autocratic strategies. 

\section{Continuous Donation Game}\label{sec:continuousDonationGame}
In the classical Donation Game, cooperators pay a cost, $c$, to provide a benefit, $b$, to the opponent \citep{sigmund:PUP:2010}. Defectors make no donations and pay no costs. The payoff matrix for this game is given by Eq. (\ref{eq:classicalPD}) with $R=b-c$, $S=-c$, $T=b$, and $P=0$. For $b>c>0$ a social dilemma arises because the payoff for mutual defection (the Nash equilibrium) is strictly less than the payoff for mutual cooperation, yet both players are tempted to shirk donations, which represents an instance of the Prisoner's Dilemma. In the iterated (and undiscounted) version of the Donation Game, the main result of \citet{press:PNAS:2012} implies that a memory-one strategy for player $X$, $\mathbf{p}=\left(p_{CC},p_{CD},p_{DC},p_{DD}\right)^{\T}$, enforces $\pi_{X}-\kappa =\chi\left(\pi_{Y}-\kappa\right)$ for some $\kappa$ and $\chi\geqslant 1$ whenever there exists a scalar, $\phi$, for which
\begin{subequations}
\label{eq:pdiscrete}
\begin{align}
p_{CC} &= 1-\phi\left(\chi -1\right)\left(b-c-\kappa\right) ; \label{eq:CCdiscrete} \\
p_{CD} &= 1-\phi\left(\chi b+c-\left(\chi -1\right)\kappa\right) ; \\
p_{DC} &= \phi\left(b+\chi c+\left(\chi -1\right)\kappa\right) ; \\
p_{DD} &= \phi\left(\chi -1\right)\kappa . \label{eq:DDdiscrete}
\end{align}
\end{subequations}
The term $\chi$ denotes the extortion factor and $\kappa$ the baseline payoff, i.e. the payoff of $\mathbf{p}$ against itself \citep{hilbe:NC:2014}. For example, if $\kappa =0$ and $\chi\geqslant 1$, then \eq{pdiscrete} defines an extortionate strategy, which unilaterally enforces $\pi_{X}=\chi\pi_{Y}$ as long as $\phi$ is sufficiently small. In this sense, $\phi$ acts as a scaling factor to ensure each coordinate of $\mathbf{p}$ falls between $0$ and $1$.

Instead of discrete ``levels" of cooperation, the continuous Donation Game admits a range of cooperation levels, $\left[0,K\right]$, with $K>0$ indicating maximal cooperation. The costs and benefits associated with $s$, denoted by $c\left(s\right)$ and $b\left(s\right)$, respectively, are nondecreasing functions of $s$ and, in analogy to the discrete case, satisfy $b\left(s\right)>c\left(s\right)$ for $s>0$ and $c(0)=b(0)=0$ \citep{killingback:PRSB:1999,wahl:JTB:1999a,wahl:JTB:1999b,killingback:AN:2002,doebeli:S:2004}. The payoff matrix, \eq{classicalPD} is replaced by payoff functions, with the payoffs to players $X$ and $Y$ for playing $x$ against $y$ being $u_{X}\left(x,y\right) :=b\left(y\right) -c\left(x\right)$ and $u_{Y}\left(x,y\right) =u_{X}\left(y,x\right) =b\left(x\right) -c\left(y\right)$, respectively (i.e. the game is symmetric). For this natural extension of the classical Donation Game, we first show the existence of autocratic and, in particular, extortionate strategies, that play only $x=0$ and $x=K$ and ignore all other cooperation levels.

\subsection{Two-point autocratic strategies}\label{subsec:twoPointStrategies}

For the continuous Donation Game, we show, using Theorem \ref{thm:mainTheorem}, that player $X$ can unilaterally enforce $\pi_{X}-\kappa =\chi\left(\pi_{Y}-\kappa\right)$ for fixed $\chi$ and $\kappa$ by playing only two actions: $x=0$ (defect) and $x=K$ (fully cooperate). Conditioned on the fact that $X$ plays only $0$ and $K$, a memory-one strategy for player $X$ is defined by a reaction function, $p\left(x,y\right)$, which denotes the probability that $X$ plays $K$ following an outcome in which $X$ plays $x\in\left\{0,K\right\}$ and $Y$ plays $y\in\left[0,K\right]$; $1-p\left(x,y\right)$ is the probability that $X$ plays $0$ (i.e. defects). Player $X$'s initial action is determined by the probability, $p_{0}$, that $X$ plays $x=K$ in the first round.

Consider the function $\psi\left(s\right) :=-\chi b\left(s\right) -c\left(s\right)$ and suppose that $0\leqslant\kappa\leqslant b\left(K\right) -c\left(K\right)$. If player $X$'s initial action is $x=K$ with probability $p_{0}$ and $x=0$ with probability $1-p_{0}$ then, for sufficiently weak discounting or, equivalently, sufficiently many rounds of interaction,
\begin{align}\label{eq:deltaConditionDonation}
\lambda &\geqslant \frac{b\left(K\right) +\chi c\left(K\right)}{\chi b\left(K\right) +c\left(K\right)} ,
\end{align}
and for $p_{0}$ falling within a feasible range (see Eq. (\ref{sieq:twoPointInitialBoundExtortion}) in \textbf{Supporting Information}), the reaction function
\begin{align}\label{eq:twoPointDonationExample}
p(x,y) &= \frac{1}{\lambda}\left(\frac{b\left(y\right) + \chi c\left(y\right) +\left(\chi -1\right)\kappa}{\chi b\left(K\right) + c\left(K\right)}-\left(1-\lambda\right) p_{0}\right) ,
\end{align}
defines a memory-one strategy,
\begin{subequations}
\begin{align}
\sigma_{X}^{0} &= \left(1-p_{0}\right)\delta_{0}+p_{0}\delta_{K} ; \\
\sigma_{X}\left[x,y\right] &= \left( 1-p\left(x,y\right) \right) \delta_{0} + p\left(x,y\right) \delta_{K} ,
\end{align}
\end{subequations}
where $\delta_{s}$ denotes the Dirac measure centered at $s\in\left[0,K\right]$, that enforces the equation $\pi_{X}-\kappa=\chi\left(\pi_{Y}-\kappa\right)$. If there is no discounting (i.e. $\lambda =1$), then the initial move is irrelevant and $p_{0}$ can be anything in the interval $\left[0,1\right]$. Note that \eq{twoPointDonationExample} represents a reactive strategy \citep{nowak:AAM:1990} because $X$ conditions her play on only the previous move of the opponent (see \fig{donationGame}(B)).

For $\kappa =0$ and $\chi\geqslant 1$, \eq{twoPointDonationExample} defines an extortionate strategy, $\sigma_{X}$, which guarantees player $X$ a fixed share of the payoffs over the payoff for mutual defection. If $\chi =1$ (and $\kappa$ is arbitrary), then this strategy is fair since player $X$ ensures the opponent has a payoff equal to her own \citep{hilbe:PNAS:2014}. On the other hand, if $\kappa =b\left(K\right) -c\left(K\right)$ and $\chi\geqslant 1$, then \eq{twoPointDonationExample} defines a generous (or ``compliant") strategy \citep{stewart:PNAS:2012,hilbe:PNAS:2013}. By playing a generous strategy, player $X$ ensures that her payoff is at most that of her opponent's. For each of these types of strategies, the probability that $X$ reacts to $y$ by cooperating is increasing as a function of $y$. In particular, $X$ is most likely to cooperate after $Y$ fully cooperates ($y=K$) and is most likely to defect after $Y$ defects ($y=0$). Moreover, this single choice of $\psi\left(s\right) =-\chi b\left(s\right) -c\left(s\right)$ demonstrates the existence of each of these three classes of autocratic strategies for the continuous Donation Game provided $\lambda$ is sufficiently weak. 

Similarly, if $\psi\left(s\right) =b\left(s\right)$ and $\lambda\geqslant c\left(K\right) /b\left(K\right)$, then
\begin{align}
p\left(x,y\right) &= \frac{1}{\lambda}\left(\frac{c\left(y\right) +\gamma}{b\left(K\right)}-\left(1-\lambda\right) p_{0}\right) ,
\end{align}
defines a reaction probability that allows $X$ to set $\pi_{Y}=\gamma$ for any $\gamma$ satisfying $0\leqslant\gamma\leqslant b\left(K\right) -c\left(K\right)$, provided $p_{0}$ falls within a suitable range (see Eq. (\ref{sieq:twoPointInitialBoundEqualizer}) in \textbf{Supporting Information}). A strategy that allows a player to single-handedly set the score of the opponent is termed an equalizer strategy \citep{boerlijst:AMM:1997,hilbe:PNAS:2013}. However, no autocratic strategy allows player $X$ to set her own score via \eq{mainEquation} (see \textbf{Supporting Information}). These results are consistent with the observations of \citet{press:PNAS:2012} that, in the classical Prisoner's Dilemma without discounting, player $X$ cannot set her own score but can set player $Y$'s score to anything between the payoffs for mutual defection and mutual cooperation.

\begin{figure*}
\includegraphics[scale=0.35]{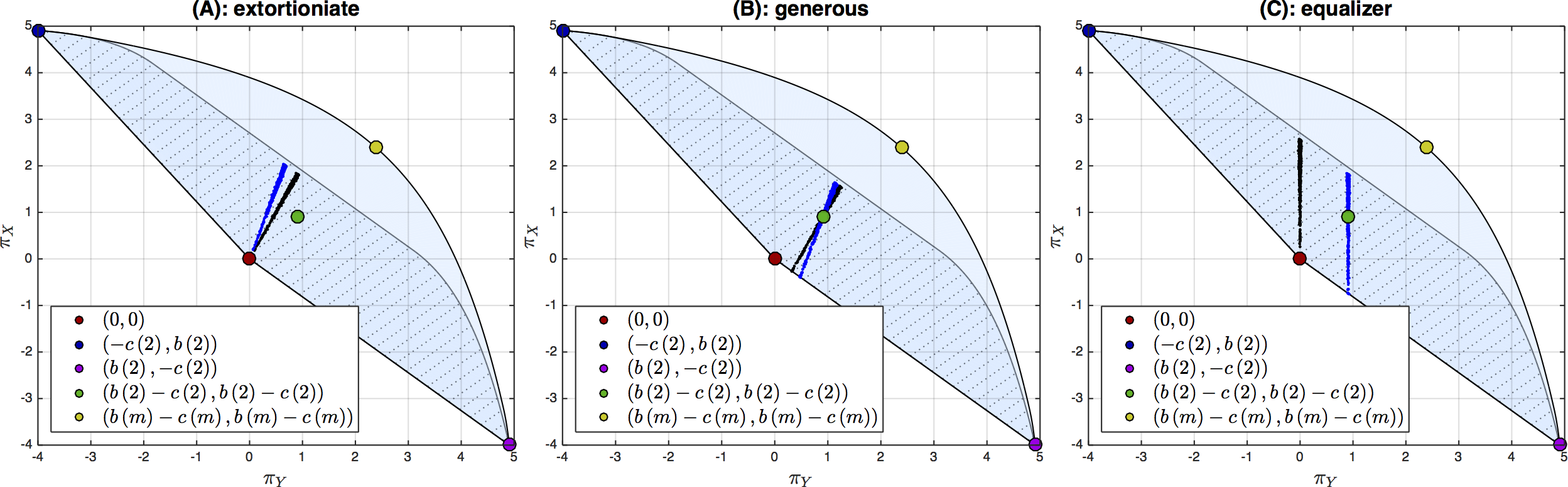}
\caption{Two-point extortionate, generous, and equalizer strategies for the continuous Donation Game with action spaces $S_{X}=\left\{0,2\right\}$ and $S_{Y}=\left[0,2\right]$ and discounting factor $\lambda =0.95$. The costs are linear, $c\left(s\right) =2s$, and the benefits saturating, $b\left(c\left(s\right)\right) =5\left(1-e^{-c\left(s\right)}\right)$. The light blue regions represent the feasible payoff pairs, $\left(\pi_{Y},\pi_{X}\right)$, for the repeated, continuous Donation Game when $X$ takes advantage of the entire action space, $\left[0,2\right]$, and the hatched regions represent the feasible payoff pairs when $X$ uses just $\left\{0,2\right\}$. In each panel, the simulation results were obtained by plotting the average payoffs for a fixed autocratic strategy against $1000$ randomly chosen memory-one strategies. (A) depicts extortionate strategies enforcing $\pi_{X}=\chi\pi_{Y}$ with $\chi =2$ (black) and $\chi =3$ (blue). (B) illustrates generous strategies enforcing $b\left(2\right) -c\left(2\right) - \pi_{X} = \chi\left(b\left(2\right) -c\left(2\right) -\pi_{Y}\right)$ with $\chi =2$ (black) and $\chi =3$ (blue). (C) demonstrates equalizer strategies enforcing $\pi_{Y}=\gamma$ with $\gamma =0$ (black) and $\gamma =b\left(2\right) -c\left(2\right)$ (blue). From (A) and (B), it is clear that neither mutual extortion nor mutual generosity is a Nash equilibrium since, when both players use extortionate or generous strategies, either player can single-handedly improve the payoffs of both players by deviating from his or her strategy. When both players together use the same action, the best outcome occurs at an investment level of $m=\frac{1}{2}\log 5<2$, which results in a payoff of $b\left(m\right) -c\left(m\right) >b\left(2\right) -c\left(2\right)$ to both players (yellow dot). \label{fig:twoPointSimulation}}
\end{figure*}

\subsection{Deterministic autocratic strategies}

One feature of two-point autocratic strategies is that they allow a player to exert control over the payoffs of a repeated game while ignoring most of the action space. One drawback is that they restrict the region of feasible game payoffs (see Fig. \ref{fig:twoPointSimulation}). This shortcoming of two-point strategies leads to a new class of strategies called deterministic strategies, which are perhaps the simplest alternatives to two-point strategies.

A deterministic strategy requires a player to respond to a history of previous play by playing an action with certainty rather than probabilistically. For example, a memory-one deterministic strategy for player $X$ is defined by (i) an initial action, $x_{0}\in S_{X}$, and (ii) a reaction function, $r^{X}:S_{X}\times S_{Y}\rightarrow S_{X}$, such that $X$ plays $r^{X}\left(x,y\right)$ following an outcome in which $X$ plays $x$ and $Y$ plays $y$. One well-known example of a deterministic strategy is tit-for-tat in the classical Prisoner's Dilemma, which is defined by $x_{0}=C$ (initially cooperate) and $r^{X}\left(x,y\right) =y$ (do what the opponent did in the previous round). Tit-for-tat is also an autocratic strategy since it enforces the fair relationship $\pi_{X}=\pi_{Y}$ \citep{stewart:PNAS:2013}.

For general memory-one deterministic strategies, the condition for the existence of autocratic strategies, Eq. (\ref{eq:mainEquation}), becomes
\begin{align}\label{eq:mainEquationDeterministic}
\alpha u_{X}\left(x,y\right) + \beta  u_{Y}\left(x,y\right) + \gamma &= \psi\left(x\right) - \lambda\psi\left(r^{X}\left(x,y\right)\right) - \left(1-\lambda\right)\psi\left(x_{0}\right) .
\end{align}
If Eq. (\ref{eq:mainEquationDeterministic}) holds for each $x\in S_{X}$ and $y\in S_{Y}$, then the deterministic strategy defined by $x_{0}$ and $r^{X}$ enforces $\alpha\pi_{X}+\beta\pi_{Y}+\gamma =0$. Thus, in order to enforce $\pi_{X}-\kappa =\chi\left(\pi_{Y}-\kappa\right)$ for $0\leqslant\kappa\leqslant b\left(K\right) -c\left(K\right)$, one may choose $\psi\left(s\right) :=-\chi b\left(s\right) -c\left(s\right)$ and use the reaction function
\begin{align}
r^{X}\left(x,y\right) &= \psi^{-1}\left(\frac{-b\left(y\right) -\chi c\left(y\right) -\left(\chi -1\right)\kappa -\left(1-\lambda\right)\psi\left(x_{0}\right)}{\lambda}\right) ,
\end{align}
where $\psi^{-1}\left(\cdots\right)$ denotes the inverse of the function $\psi$, provided $\lambda$ satisfies Eq. (\ref{eq:deltaConditionDonation}) and $x_{0}$ falls within a feasible range (see Eq. (\ref{sieq:deterministicInitialBoundExtortion}) in \textbf{Supporting Information}).

Similarly, to enforce $\pi_{Y}=\gamma$ for $0\leqslant\gamma\leqslant b\left(K\right) -c\left(K\right)$, one may choose $\psi\left(s\right) :=b\left(s\right)$ and use the reaction function
\begin{align}
r^{X}\left(x,y\right) &= \psi^{-1}\left(\frac{c\left(y\right) +\gamma -\left(1-\lambda\right)\psi\left(x_{0}\right)}{\lambda}\right) ,
\end{align}
provided $\lambda\geqslant c\left(K\right) /b\left(K\right)$ and, again, $x_{0}$ falls within a feasible range (see Eq. (\ref{sieq:deterministicInitialBoundEqualizer}) in \textbf{Supporting Information}).

Examples of deterministic extortionate, generous, and equalizer strategies are given in Fig. \ref{fig:deterministicSimulation}. It is evident that deterministic strategies increase the feasible region in which linear payoff relationships can be enforced as compared to their two-point counterparts (c.f. Fig. \ref{fig:twoPointSimulation}).

\begin{figure*}
\includegraphics[scale=0.35]{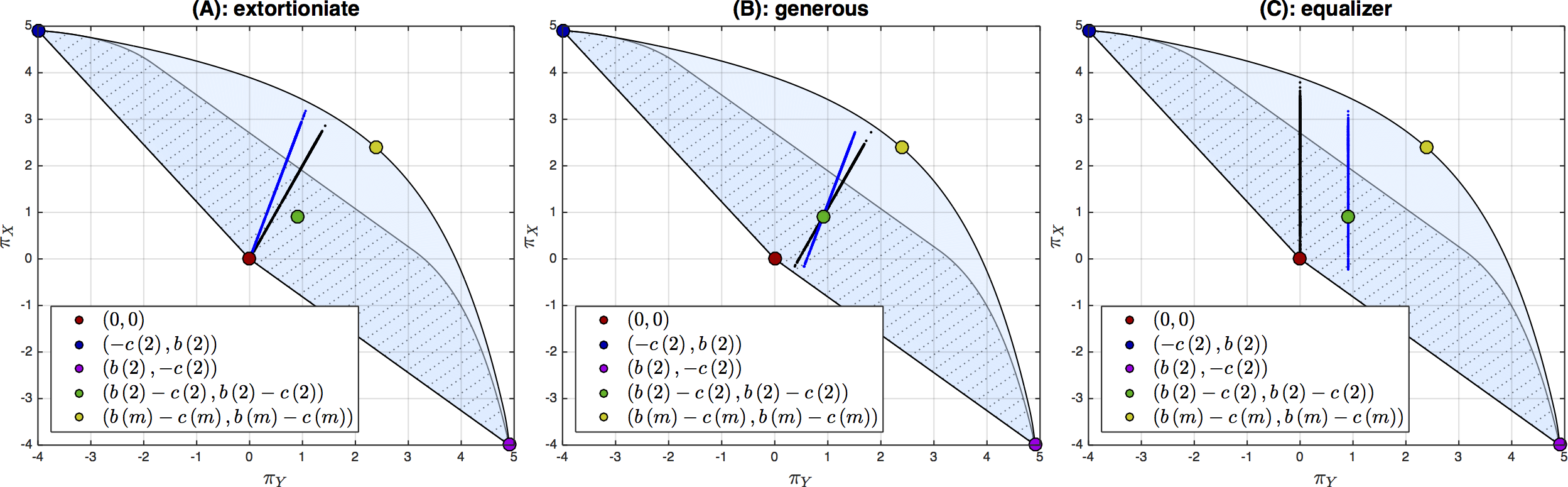}
\caption{Deterministic extortionate, generous, and equalizer strategies for the continuous Donation Game with action spaces $S_{X}=S_{Y}=\left[0,2\right]$ and discounting factor $\lambda =0.95$. As in Fig. \ref{fig:twoPointSimulation}, the simulation results in each panel were obtained by plotting the average payoffs for a fixed autocratic strategy against $1000$ randomly chosen memory-one strategies. (A) and (B) demonstrate extortionate and generous strategies, respectively, with $\chi =2$ (black) and $\chi =3$ (blue). (C) shows equalizer strategies enforcing $\pi_{Y}=\gamma$ with $\gamma =0$ (black) and $\gamma =b\left(2\right) -c\left(2\right)$ (blue). Since deterministic strategies allow player $X$ to use a much larger portion of the action space than just $0$ and $K$, we observe that the linear relationships enforced by deterministic autocratic strategies cover a greater portion of the feasible region than do two-point strategies. The best outcome when both players use the same action occurs at $m=\frac{1}{2}\log 5<2$, which results in a mutual payoff of $b\left(m\right) -c\left(m\right) >b\left(2\right) -c\left(2\right)$ (yellow dot). \label{fig:deterministicSimulation}}
\end{figure*}

\section{Discussion}
In games with two actions, zero-determinant strategies are typically defined via a technical condition such as \eq{pressDysonVector} \citep{press:PNAS:2012,hilbe:GEB:2015}. This definition makes generalizations to games with larger action spaces difficult because \eq{pressDysonVector} makes sense only for two-action games. Therefore, we introduce the more general term autocratic strategy for any strategy that unilaterally enforces a linear relationship on expected payoffs. Of course, this linear relationship is precisely what makes strategies satisfying \eq{pressDysonVector} interesting.

\thm{mainTheorem} provides a condition for the existence of autocratic strategies for games with general action spaces.  We illustrate this phenomenon with a continuous-action-space extension of the classical Donation Game, which represents an instance of the Prisoner's Dilemma. The existing literature on zero-determinant strategies for the classical Prisoner's Dilemma provides no way of treating this continuous extension of the Donation Game. However, \thm{mainTheorem} makes no assumptions on the action space of the game and thus applies to the continuous Donation Game as well as its classical counterpart.

Surprisingly, in many cases a player can enforce a linear relationship on expected payoffs by playing only two actions, despite the fact that the opponent may have infinitely many actions available to use (Corollary \ref{cor:mainCorollary} in \textbf{Supporting Information}). We demonstrate that the conditions guaranteeing the existence of extortionate, generous, fair, and equalizer strategies in the continuous Donation Game are in fact similar to those of the two-action case. However, despite the simplicity of these two-point strategies, a player needs to know how to respond to every possible move of the opponent; knowledge of how to respond to just defection ($y=0$) and full cooperation ($y=K$) does not suffice. Therefore, although a player using a two-point strategy plays only $x=0$ and $x=K$, these strategies represent a departure from the classical Donation Game.

Another important difference is that, whereas in the classical Prisoner's Dilemma mutual generosity represents a symmetric Nash equilibrium \citep{hilbe:GEB:2015}, this need not be the case in the continuous Donation Game. Instead, mutual generosity results in a payoff of $b\left(K\right) -c\left(K\right)$ for each player, where $K$ is the maximal investment, but intermediate levels of cooperation may yield $b\left(m\right) -c\left(m\right) >b\left(K\right) -c\left(K\right)$ for $m\in\left(0,K\right)$, i.e. both players fare better if they each invest $m$ instead of $K$ (see Figs. \ref{fig:twoPointSimulation} and \ref{fig:deterministicSimulation}). However, no player can enforce a generous relationship with baseline payoff $\kappa =b\left(m\right) -c\left(m\right)$ because it is outside of the feasible range for $\kappa$ (see \textbf{Supporting Information}). Thus, the performance of a generous strategy as a response to itself depends critically on whether the game has two actions or a continuous range of actions.

Extortionate strategies for the iterated Prisoner's Dilemma are not evolutionarily stable \citep{adami:NC:2013}. Since mutual generosity in the continuous Donation Game need not be a Nash equilibrium, it follows that generous strategies also need not be evolutionarily stable. Moreover, against human opponents, extortioners are punished by a refusal to fully cooperate, while generous players provide their opponents with an incentive to cooperate and fare better in experiments \citep{hilbe:NC:2014}. Such behavior supports what one would expect from a player using a fair autocratic strategy enforcing $\pi_{X}=\pi_{Y}$ (such as tit-for-tat): if the opponent ensures $\pi_{X}-\kappa =\chi\left(\pi_{Y}-\kappa\right)$ for some $\chi >1$, then both players get $\kappa$. In particular, fair strategies punish extortion and reward generosity. Based on our results on generous strategies for the continuous Donation Game, it would be interesting to see whether (human) experiments support the same conclusion and whether the participants succeed in securing payoffs that exceed those of mutual generosity. The performance of autocratic strategies in populations, however, is but one perspective on this recently-discovered class of strategies for repeated games.

In games with two discrete actions, our definition of a Press-Dyson function specializes to a multiple of the Press-Dyson vector, $C\widetilde{\mathbf{p}}$ with $C$ a constant (see \eq{pressDysonVector} and \textbf{Supporting Information} for details). The Press-Dyson vector is recovered by normalizing the Press-Dyson function and thus eliminating the constant, $C$. However, in games with $d$ actions, this function involves at least $d-1$ constants, which, for $d>2$, cannot be eliminated by normalization (see \textbf{Supporting Information} for an example with three actions). Therefore, based on \thm{mainTheorem}, in two-action games it is perhaps more appropriate to define a Press-Dyson vector to be any vector of the form $C\widetilde{\mathbf{p}}$. This distinction for games with two actions is minor, however, and does not change the fact that Theorem $\ref{thm:mainTheorem}$ covers all of the known results on the existence of zero-determinant strategies for repeated games.

More importantly, however, the analysis of iterated games with only two actions completely misses the fact that autocratic strategies are most naturally presented implicitly via \eq{mainEquation}. Even in the case of two actions, infinitely many autocratic strategies may exist \citep{press:PNAS:2012}, but their simplistic nature admits explicit solutions. Our extension shows that, in general, (i) autocratic strategies need not be unique and (ii) one cannot explicitly list all autocratic strategies that produce a fixed Press-Dyson function. Thus, for arbitrary action spaces (but already for games with $d>2$ actions), the space of autocratic strategies is more sophisticated than two-action games suggest. Notwithstanding the intrinsic difficulty in explicitly specifying all autocratic strategies, our results demonstrate that these strategies exist in a broad class of games and are not simply consequences of the finiteness of the action space in games such as the Prisoner's Dilemma.

\section*{Acknowledgments}

The authors are grateful to Christian Hilbe for reading an earlier draft of this manuscript and for suggesting that we look into deterministic strategies. The authors acknowledge financial support from the Natural Sciences and Engineering Research Council of Canada (NSERC), grant RGPIN-2015-05795.

\bibliographystyle{plainnat}

\newpage

\setcounter{section}{0}
\setcounter{equation}{0}
\renewcommand{\thesection}{SI.\arabic{section}}
\renewcommand{\theequation}{SI.\arabic{equation}}

\section*{Supporting Information}

\section{Iterated games with two players and measurable action spaces}\label{subsec:iteratedGames}
By ``action space," we mean a measurable space, $S$, equipped with a $\sigma$-algebra, $\mathcal{F}\left(S\right)$ (although we suppress $\mathcal{F}\left(S\right)$ and refer to the space simply as $S$). Informally, $S$ is the space of actions, decisions, investments, or options available to a player at each round of the iterated interaction and could be a finite set, a continuous interval, or something more complicated. Since the players need not have the same action space, we denote by $S_{X}$ the space of actions available to player $X$ and by $S_{Y}$ the space of actions available to player $Y$. In what follows, all functions are assumed to be measurable and bounded.

In each encounter (i.e. ``one-shot game"), the players receive payoffs based on a payoff function,
\begin{linenomath}
\begin{align}
u = \left(u_{X},u_{Y}\right) &: S_{X}\times S_{Y}\longrightarrow\mathbb{R}^{2} .
\end{align}
\end{linenomath}
The first and second coordinate functions, $u_{X}$ and $u_{Y}$, give the payoffs to players $X$ and $Y$, respectively. An iterated game between players $X$ and $Y$ consists of a sequence of these one-shot interactions. If, at time $t$, player $X$ uses $x_{t}\in S_{X}$ and player $Y$ uses $y_{t}\in S_{Y}$, then the (normalized) payoff to player $X$ for a sequence of $T+1$ interactions (from time $t=0$ to $t=T$) is
\begin{linenomath}
\begin{align}\label{eq:totalPayoff}
\frac{1-\lambda}{1-\lambda^{T+1}}\sum_{t=0}^{T}\lambda^{t}u_{X}\left(x_{t},y_{t}\right) ,
\end{align}
\end{linenomath}
where $\lambda$ is the discounting factor, $0<\lambda <1$. The payoff to player $Y$ is obtained by replacing $u_{X}$ by $u_{Y}$ in Eq. (\ref{eq:totalPayoff}). Thus, the discounted payoffs, $\lambda^{t}u_{X}\left(x_{t},y_{t}\right)$, are simply added up and then normalized by a factor of $\frac{1-\lambda}{1-\lambda^{T+1}}$ to ensure that the payoff for the repeated game is measured in the same units as the payoffs for individual encounters \citep{fudenberg:MIT:1991}. Moreover, provided the series $\sum_{t=0}^{\infty}u_{X}\left(x_{t},y_{t}\right)$ is Ces\`{a}ro summable, meaning $\lim_{T\rightarrow\infty}\frac{1}{T+1}\sum_{t=0}^{T}u_{X}\left(x_{t},y_{t}\right)$ exists, we have
\begin{linenomath}
\begin{align}
\lim_{\lambda\rightarrow 1^{-}}\left(1-\lambda\right)\sum_{t=0}^{\infty}\lambda^{t}u_{X}\left(x_{t},y_{t}\right) &= \lim_{T\rightarrow\infty}\frac{1}{T+1}\sum_{t=0}^{T}u_{X}\left(x_{t},y_{t}\right)
\end{align}
\end{linenomath}
\citep[see][]{korevaar:S:2004}. Therefore, payoffs in the undiscounted case may be obtained from the payoffs for discounted games in the limit $\lambda\rightarrow 1^{-}$, provided this limit exists \citep[see][]{hilbe:GEB:2015}.

Here, we consider stochastic strategies that condition on the history of play: both players observe the sequence of play up to the current period and use it to devise an action for the present encounter. In order to formally define such strategies, we first recall the notion of ``history" in a repeated game: a history at time $T$ is a sequence of action pairs,
\begin{linenomath}
\begin{align}
h^{T} &:= \Big(\left(x_{0},y_{0}\right) , \dots , \left(x_{T-1},y_{T-1}\right)\Big) ,
\end{align}
\end{linenomath}
indicating the sequence of play leading up to the interaction at time $T$ \citep{fudenberg:MIT:1991}. In other words, a history at time $T$ is an element of $\mathcal{H}^{T}:=\prod_{t=0}^{T-1}S_{X}\times S_{Y}$. Let $\mathcal{H}^{0}:=\left\{\varnothing\right\}$, where $\varnothing$ denotes the ``empty" history (which serves just to indicate that there has been no history of past play, i.e. that the game has not yet begun). The set of all possible histories is the disjoint union, $\mathcal{H}:=\bigsqcup_{T\geqslant 0}\mathcal{H}^{T}$. For $h^{T}=\Big(\left(x_{0},y_{0}\right) , \dots , \left(x_{T-1},y_{T-1}\right)\Big)$ and $t\leqslant T-1$, let
\begin{linenomath}
\begin{subequations}
\begin{align}
h_{t}^{T} &:= \left(x_{t},y_{t}\right) ; \\
h_{\leqslant t}^{T} &:= \Big( \left(x_{0},y_{0}\right) , \dots , \left(x_{t},y_{t}\right) \Big) .
\end{align}
\end{subequations}
\end{linenomath}
That is, $h_{t}^{T}$ is the action pair played at time $t$, and $h_{\leqslant t}^{T}$ is the ``sub-history" of $h^{T}$ until time $t\leqslant T$.

A pure strategy for player $X$ in the repeated game is a map, $\mathcal{H}\rightarrow S_{X}$, indicating an action in $S_{X}$ (deterministically) for each history leading up to the current encounter. More generally, $X$ could look at the history of past play and use this information to choose an action from $S_{X}$ probabilistically (rather than deterministically). A strategy of this form is known as a behavioral strategy \citep{fudenberg:MIT:1991}. In terms of $\mathcal{H}$, a behavioral strategy for player $X$ is a map
\begin{linenomath}
\begin{align}
\sigma_{X} &: \mathcal{H} \longrightarrow \Delta\left(S_{X}\right) ,
\end{align}
\end{linenomath}
where $\Delta\left(S_{X}\right)$ is the space of probability measures on $S_{X}$. An important type of behavioral strategy is a Markov strategy, which is a behavioral strategy, $\sigma_{X}$, that satisfies $\sigma_{X}\left[h^{T}\right] =\sigma_{X}\left[h^{T}_{T-1}\right]$.
That is, a Markov strategy depends on only the last pair of actions and not on the entire history of play. Note, however, that a Markov strategy may still depend on $t$. If $\sigma_{X}$ is a Markov strategy that does not depend on $t$, then we say that $\sigma_{X}$ is a stationary (or memory-one) strategy.

Suppose that $\sigma_{X}$ and $\sigma_{Y}$ are behavioral strategies for players $X$ and $Y$, respectively. Consider the map, $\sigma$, defined by the product measure,
\begin{linenomath}
\begin{align}\label{eq:kappaDefinition}
\sigma := \sigma_{X}\times\sigma_{Y} &: \mathcal{H} \longrightarrow \Delta\left(S_{X}\times S_{Y}\right) \nonumber \\
&: h^{t} \longmapsto \sigma_{X}\left[h^{t}\right]\times\sigma_{Y}\left[h^{t}\right] .
\end{align}
\end{linenomath}
By the Hahn-Kolmogorov theorem, for each $t\geqslant 0$ there exists a unique measure, $\mu_{t}$, on $\mathcal{H}^{t+1}$ such that for each $E'\in\mathcal{F}\left(\mathcal{H}^{t}\right)$ and $E\in\mathcal{F}\left(S_{X}\times S_{Y}\right)$,
\begin{linenomath}
\begin{align}
\mu_{t}\left(E'\times E\right) &= \int_{h^{t}\in E'} \sigma\left(h^{t},E\right) \,d\sigma\left(h_{\leqslant t-2}^{t},h_{t-1}^{t}\right)\cdots\,d\sigma\left(h_{\leqslant 0}^{t},h_{1}^{t}\right)\,d\sigma\left(\varnothing, h_{0}^{t}\right) ,
\end{align}
\end{linenomath}
where, for $h\in\mathcal{H}$ and $\mathbf{s}\in S_{X}\times S_{Y}$, $d\sigma\left(h,\mathbf{s}\right)$ denotes the differential of the measure $\sigma\left(h,\--\right)$ on $S_{X}\times S_{Y}$. In the case $t=0$, this measure is simply the product of the two initial actions, i.e. $\mu_{0}=\sigma_{X}\left[\varnothing\right]\times\sigma_{Y}\left[\varnothing\right]$. From these measures, we obtain a sequence of measures, $\left\{\nu_{t}\right\}_{t\geqslant 0}\subseteq\Delta\left(S_{X}\times S_{Y}\right)$, defined by
\begin{linenomath}
\begin{align}\label{eq:nuDefinition}
\nu_{t}\left(E\right) &:= \mu_{t}\left(\mathcal{H}^{t}\times E\right) .
\end{align}
\end{linenomath}
Informally, $\nu_{t}\left(E\right)$ is the probability that the action pair at time $t$ is in $E\subseteq S_{X}\times S_{Y}$, averaged over all histories that lead to $E$.
The sequences, $\left\{\mu_{t}\right\}_{t\geqslant 0}$ and $\left\{\nu_{t}\right\}_{t\geqslant 0}$, admit a convenient format for the expected payoffs, $\pi_{X}^{T}$ and $\pi_{Y}^{T}$, to players $X$ and $Y$, respectively. Before stating this result, we first formally define expected payoffs for the $\left(T+1\right)$-period game (where $T<\infty$):
\begin{definition}[Objective function for a finite game]
If $\sigma_{X}$ and $\sigma_{Y}$ are behavioral strategies for players $X$ and $Y$, respectively, and if $\sigma =\sigma_{X}\times\sigma_{Y}$ (see Eq. (\ref{eq:kappaDefinition})), then the objective function (or expected payoff) for player $X$ in the $\left(T+1\right)$-period game is
\begin{linenomath}
\begin{align}
\pi_{X}^{T} &:= \int_{h^{T+1}\in\mathcal{H}^{T+1}}\left[\frac{1-\lambda}{1-\lambda^{T+1}}\sum_{t=0}^{T}\lambda^{t}u_{X}\left(h_{t}^{T+1}\right)\right] \,d\sigma\left(h_{\leqslant T-1}^{T+1},h_{T}^{T+1}\right)\cdots\,d\sigma\left(h_{\leqslant 0}^{T+1},h_{1}^{T+1}\right)\,d\sigma\left(\varnothing, h_{0}^{T+1}\right) .
\end{align}
\end{linenomath}
\end{definition}

Using $\left\{\nu_{t}\right\}_{t\geqslant 0}$, we can write $\pi_{X}^{T}$ differently:
\begin{lemma}\label{lem:expectedPayoff}
For fixed $\sigma_{X}$ and $\sigma_{Y}$ generating $\left\{\nu_{t}\right\}_{t\geqslant 0}$, we have
\begin{linenomath}
\begin{align}
\pi_{X}^{T} &= \frac{1-\lambda}{1-\lambda^{T+1}}\sum_{t=0}^{T}\lambda^{t}\int_{\left(x,y\right)\in S_{X}\times S_{Y}} u_{X}\left(x,y\right) \,d\nu_{t}\left(x,y\right) .
\end{align}
\end{linenomath}
\end{lemma}
As a consequence of Lemma \ref{lem:expectedPayoff}, we see that $\lim_{T\rightarrow\infty}\pi_{X}^{T}$ exists since $\nu_{t}$ is a probability measure and
\begin{linenomath}
\begin{align}
\left|\int_{\left(x,y\right)\in S_{X}\times S_{Y}} u_{X}\left(x,y\right) \,d\nu_{t}\left(x,y\right)\right| &\leqslant \sup_{\left(x,y\right)\in S_{X}\times S_{Y}}\left| u_{X}\left(x,y\right) \right| < \infty
\end{align}
\end{linenomath}
by the fact that $u_{X}$ is bounded. Thus, we define the objective function for an infinite game as follows:
\begin{definition}[Objective function for an infinite game]
If $\sigma_{X}$ and $\sigma_{Y}$ are behavioral strategies for players $X$ and $Y$, respectively, and if $\sigma =\sigma_{X}\times\sigma_{Y}$, then the objective function for player $X$ in the infinite game is
\begin{linenomath}
\begin{align}\label{eq:expectedPayoff}
\pi_{X} &:= \lim_{T\rightarrow\infty}\pi_{X}^{T} = \left(1-\lambda\right)\sum_{t=0}^{\infty}\lambda^{t}\int_{\left(x,y\right)\in S_{X}\times S_{Y}} u_{X}\left(x,y\right) \,d\nu_{t}\left(x,y\right) .
\end{align}
\end{linenomath}
\end{definition}

\begin{remark}
Classically, the objective function of a repeated game with infinitely many rounds is defined using a distribution over ``infinite histories," which is generated by the players' strategies for the repeated game \citep{fudenberg:MIT:1991}. That is, for $\mathcal{H}^{\infty}:=\prod_{t=0}^{\infty}S_{X}\times S_{Y}$ and some measure, $\mu\in\Delta\left(\mathcal{H}^{\infty}\right)$, the objective function of player $X$ is defined by
\begin{linenomath}
\begin{align}\label{eq:classicalPayoff}
\int_{h^{\infty}\in\mathcal{H}^{\infty}}\left(1-\lambda\right)\sum_{t=0}^{\infty}\lambda^{t}u_{X}\left(h_{t}^{\infty}\right) \, d\mu\left(h^{\infty}\right) .
\end{align}
\end{linenomath}
Using Eq. (\ref{eq:expectedPayoff}) as an objective function for player $X$, we do not need to worry about what $\mu$ is (or if it even exists for a general action space) since Eq. (\ref{eq:classicalPayoff}), whenever it is defined, must coincide with Eq. (\ref{eq:expectedPayoff}). To see why, suppose that there is a distribution, $\mu$, on $\mathcal{H}^{\infty}$ that satisfies
\begin{linenomath}
\begin{align}\label{eq:reductionOnCylinders}
\mu\left(E\times \left(S_{X}\times S_{Y}\right)\times\left(S_{X}\times S_{Y}\right)\times\cdots\right) &= \mu_{T}\left(E\right)
\end{align}
\end{linenomath}
for each $E\in\mathcal{F}\left(\mathcal{H}^{T+1}\right)$. Then, by the dominated convergence theorem, Eq. (\ref{eq:reductionOnCylinders}), and Eq. (\ref{eq:nuDefinition}),
\begin{linenomath}
\begin{align}
\int_{h^{\infty}\in\mathcal{H}^{\infty}} &\left(1-\lambda\right)\sum_{t=0}^{\infty}\lambda^{t}u_{X}\left(h_{t}^{\infty}\right) \, d\mu\left(h^{\infty}\right) \nonumber \\
&= \left(1-\lambda\right)\sum_{t=0}^{\infty}\lambda^{t}\int_{h^{\infty}\in\mathcal{H}^{\infty}} u_{X}\left(h_{t}^{\infty}\right) \, d\mu\left(h^{\infty}\right) \nonumber \\
&= \left(1-\lambda\right)\sum_{t=0}^{\infty}\lambda^{t}\int_{h^{t+1}\in\mathcal{H}^{t+1}} u_{X}\left(h_{t}^{t+1}\right) \, d\mu_{t}\left(h^{t+1}\right) \nonumber \\
&= \left(1-\lambda\right)\sum_{t=0}^{\infty}\lambda^{t}\int_{\left(x,y\right)\in S_{X}\times S_{Y}} u_{X}\left(x,y\right) \, d\nu_{t}\left(x,y\right) .
\end{align}
\end{linenomath}
Therefore, assuming Lemma \ref{lem:expectedPayoff}, the objective function for player $X$ defined by $\pi_{X}:=\lim_{T\rightarrow\infty}\pi_{X}^{T}$ is the same as the classical objective function for repeated games when the players' strategies produce a probability distribution over infinite histories. Typically, the existence of such a distribution depends on $S$ being finite or the measures in $\left\{\mu_{t}\right\}_{t\geqslant 0}$ being inner regular (which allows one to deduce the existence of $\mu$ from $\left\{\mu_{t}\right\}_{t\geqslant 0}$ using the Kolmogorov extension theorem). In practice, these assumptions are often not unreasonable, but with Lemma~\ref{lem:expectedPayoff}, we do not need to worry about the existence of such a distribution.
\end{remark}

In order to prove Lemma \ref{lem:expectedPayoff}, we first need a simple technical result:
\begin{lemma}\label{lem:technicalLemma}
Suppose that $\mathcal{X}$ and $\mathcal{Y}$ are measure spaces and $\sigma$ is a Markov kernel from $\mathcal{X}$ to $\mathcal{Y}$. Let $\mu$ be a probability measure on $\mathcal{X}$ and consider the measure on $\mathcal{Y}$ defined, for $E\in\mathcal{F}\left(\mathcal{Y}\right)$, by
\begin{linenomath}
\begin{align}
\nu\left(E\right) &:= \int_{x\in\mathcal{X}}\sigma\left(x,E\right)\,d\mu\left(x\right) .
\end{align}
\end{linenomath}
For any bounded, measurable function, $f:\mathcal{Y}\rightarrow\mathbb{R}$, we have
\begin{linenomath}
\begin{align}
\int_{y\in\mathcal{Y}} f\left(y\right) \, d\nu\left(y\right) &= \int_{x\in\mathcal{X}}\int_{y\in\mathcal{Y}}f\left(y\right)\,d\sigma\left(x,y\right)\,d\mu\left(x\right) ,
\end{align}
\end{linenomath}
where, for each $x\in\mathcal{X}$, $d\sigma\left(x,y\right)$ denotes the differential of the measure $\sigma\left(x,\--\right)$ on $\mathcal{Y}$.
\end{lemma}
\begin{proof}
Since $f$ is bounded, there exists a sequence of simple functions, $\left\{f_{n}\right\}_{n\geqslant 1}$, such that $f_{n}\rightarrow f$ uniformly on $\mathcal{Y}$. For each $n\geqslant 1$, let $f_{n}=\sum_{i=1}^{N_{n}}c_{i}^{n}\chi_{E_{i}^{n}}$ for some $c_{i}^{n}\in\mathbb{R}$ and $E_{i}^{n}\in\mathcal{F}\left(\mathcal{Y}\right)$, where $\chi_{E_{i}^{n}}$ is the characteristic function of $E_{i}^{n}$ (meaning $\chi_{E_{i}^{n}}\left(x\right) =1$ if $x\in E_{i}^{n}$ and $\chi_{E_{i}^{n}}\left(x\right) =0$ if $x\not\in E_{i}^{n}$). By uniform convergence,
\begin{linenomath}
\begin{align}
\int_{y\in\mathcal{Y}} f\left(y\right) \, d\nu\left(y\right) &= \lim_{n\rightarrow\infty}\int_{y\in\mathcal{Y}} f_{n}\left(y\right) \, d\nu\left(y\right) \nonumber \\
&= \lim_{n\rightarrow\infty} \sum_{i=1}^{N_{n}}c_{i}^{n}\nu\left(E_{i}^{n}\right) \nonumber \\
&= \lim_{n\rightarrow\infty} \sum_{i=1}^{N_{n}}c_{i}^{n}\int_{x\in\mathcal{X}}\sigma\left(x,E_{i}^{n}\right)\,d\mu\left(x\right) \nonumber \\
&= \lim_{n\rightarrow\infty} \int_{x\in\mathcal{X}}\sum_{i=1}^{N_{n}}c_{i}^{n}\sigma\left(x,E_{i}^{n}\right)\,d\mu\left(x\right) \nonumber \\
&= \lim_{n\rightarrow\infty} \int_{x\in\mathcal{X}}\int_{y\in\mathcal{Y}}f_{n}\left(y\right)\,d\sigma\left(x,y\right)\,d\mu\left(x\right) \nonumber \\
&= \int_{x\in\mathcal{X}}\int_{y\in\mathcal{Y}}f\left(y\right)\,d\sigma\left(x,y\right)\,d\mu\left(x\right) ,
\end{align}
\end{linenomath}
which completes the proof.
\end{proof}

\begin{proof}[Proof of Lemma \ref{lem:expectedPayoff}]
By Lemma \ref{lem:technicalLemma} and the definitions of $\mu_{t}$ and $\nu_{t}$,
\begin{linenomath}
\begin{align}
\pi_{X}^{T} &= \int_{h^{T+1}\in\mathcal{H}^{T+1}}\left[\frac{1-\lambda}{1-\lambda^{T+1}}\sum_{t=0}^{T}\lambda^{t}u_{X}\left(h_{t}^{T+1}\right)\right] \,d\sigma\left(h_{\leqslant T-1}^{T+1},h_{T}^{T+1}\right)\cdots\,d\sigma\left(h_{\leqslant 0}^{T+1},h_{1}^{T+1}\right)\,d\sigma\left(\varnothing, h_{0}^{T+1}\right) \nonumber \\
&= \frac{1-\lambda}{1-\lambda^{T+1}}\sum_{t=0}^{T}\lambda^{t}\int_{h^{T+1}\in\mathcal{H}^{T+1}} u_{X}\left(h_{t}^{T+1}\right) \,d\sigma\left(h_{\leqslant T-1}^{T+1},h_{T}^{T+1}\right)\cdots\,d\sigma\left(h_{\leqslant 0}^{T+1},h_{1}^{T+1}\right)\,d\sigma\left(\varnothing, h_{0}^{T+1}\right) \nonumber \\
&= \frac{1-\lambda}{1-\lambda^{T+1}}\sum_{t=0}^{T}\lambda^{t}\int_{h^{T+1}\in\mathcal{H}^{T+1}} u_{X}\left(h_{t}^{T+1}\right) \,d\mu_{T}\left(h^{T+1}\right) \nonumber \\
&= \frac{1-\lambda}{1-\lambda^{T+1}}\sum_{t=0}^{T}\lambda^{t}\int_{h^{t+1}\in\mathcal{H}^{t+1}} u_{X}\left(h_{t}^{t+1}\right) \,d\mu_{t}\left(h^{t+1}\right) \nonumber \\
&= \frac{1-\lambda}{1-\lambda^{T+1}}\sum_{t=0}^{T}\lambda^{t}\int_{h_{t}^{t+1}\in S_{X}\times S_{Y}} u_{X}\left(h_{t}^{t+1}\right) \,d\nu_{t}\left(h_{t}^{t+1}\right) \nonumber \\
&= \frac{1-\lambda}{1-\lambda^{T+1}}\sum_{t=0}^{T}\lambda^{t}\int_{\left(x,y\right)\in S_{X}\times S_{Y}} u_{X}\left(x,y\right) \,d\nu_{t}\left(x,y\right) ,
\end{align}
\end{linenomath}
which completes the proof.
\end{proof}

The objective function of Eq. (\ref{eq:expectedPayoff}), which is obtained using Lemma~\ref{lem:expectedPayoff}, eliminates the need to deal with histories when proving our main results for iterated games. With the background on expected payoffs now established, we turn our attention to the proofs of the results claimed in the main text:

\section{Detailed proofs of the main results}\label{subsec:detailedProofs}

Before proving our main results, we state a technical lemma that generalizes Lemma 3.1 of \citet{akin:Games:2015}--which \citet{hilbe:PNAS:2014} refer to as Akin's Lemma--and Lemma 1 of \citet{hilbe:GEB:2015}. This lemma relates the strategies of the two players, $\sigma_{X}$ and $\sigma_{Y}$, and the (discounted) sequence of play to the initial action of player $X$ when $\sigma_{X}$ is memory-one. Our proof of this lemma is essentially the same as theirs but in the broader setting of a measurable action space:
\begin{lemma}\label{lem:akinsLemma}
For any memory-one strategy, $\sigma_{X}$, and $E\in\mathcal{F}\left(S_{X}\right)$,
\begin{linenomath}
\begin{align}
\sum_{t=0}^{\infty}\lambda^{t}\int_{\left(x,y\right)\in S_{X}\times S_{Y}} \Big[ \chi_{E}\left(x\right) - \lambda\sigma_{X}\left[x,y\right]\left(E\right) \Big] \,d\nu_{t}\left(x,y\right) &= \sigma_{X}^{0}\left(E\right) ,
\end{align}
\end{linenomath}
where $\sigma_{X}^{0}$ is the initial action of player $X$.
\end{lemma}
\begin{proof}
By the definition of $\nu_{t}$, we have
\begin{linenomath}
\begin{subequations}
\begin{align}
\int_{\left(x,y\right)\in S_{X}\times S_{Y}} \chi_{E}\left(x\right) \,d\nu_{t}\left(x,y\right) &= \nu_{t}\left(E\times S\right) ; \\
\int_{\left(x,y\right)\in S_{X}\times S_{Y}} \sigma_{X}\left[x,y\right]\left(E\right) \,d\nu_{t}\left(x,y\right) &= \nu_{t+1}\left(E\times S\right) .
\end{align}
\end{subequations}
\end{linenomath}
Therefore, since $\nu_{t}$ is a probability measure (in particular, at most $1$ on any measurable set) for each $t$,
\begin{linenomath}
\begin{align}
\sum_{t=0}^{\infty}\lambda^{t} &\int_{\left(x,y\right)\in S_{X}\times S_{Y}} \Big[ \chi_{E}\left(x\right) - \lambda\sigma_{X}\left[x,y\right]\left(E\right) \Big] \,d\nu_{t}\left(x,y\right) \nonumber \\
&= \sum_{t=0}^{\infty}\lambda^{t}\Big( \nu_{t}\left(E\times S\right) - \lambda\nu_{t+1}\left(E\times S\right) \Big) \nonumber \\
&= \nu_{0}\left(E\times S\right) - \lim_{t\rightarrow\infty}\lambda^{t+1}\nu_{t+1}\left(E\times S\right) \nonumber \\
&= \nu_{0}\left(E\times S\right) \nonumber \\
&= \sigma_{X}^{0}\left(E\right) ,
\end{align}
\end{linenomath}
which completes the proof.
\end{proof}

By the definitions of $\pi_{X}$ and $\pi_{Y}$, we have
\begin{linenomath}
\begin{align}
\alpha\pi_{X} + \beta\pi_{Y} + \gamma &= \left(1-\lambda\right)\sum_{t=0}^{\infty}\lambda^{t}\int_{\left(x,y\right)\in S_{X}\times S_{Y}} \Big[ \alpha u_{X}\left(x,y\right) + \beta u_{Y}\left(x,y\right) + \gamma \Big] \, d\nu_{t}\left(x,y\right) .
\end{align}
\end{linenomath}
Since our goal is to establish Theorem \ref{thm:mainTheorem}, which states that player $X$ can enforce the relation $\alpha\pi_{X} + \beta\pi_{Y} + \gamma =0$ using some $\sigma_{X}\left[x,y\right]$ and $\sigma_{X}^{0}$, as a first step we show that $\sum_{t=0}^{\infty}\lambda^{t}\int_{\left(x,y\right)\in S_{X}\times S_{Y}} \varphi\left(x,y\right) \, d\nu_{t}\left(x,y\right) =\int_{s\in S_{X}}\psi\left(s\right) \, d\sigma_{X}^{0}\left(s\right)$ for a particular choice of $\varphi\left(x,y\right)$. We then deduce Theorem \ref{thm:mainTheorem} by setting
\begin{linenomath}
\begin{align}
\alpha u_{X}\left(x,y\right) +\beta u_{Y}\left(x,y\right) +\gamma +\left(1-\lambda\right)\int_{s\in S_{X}}\psi\left(s\right)\,\sigma_{X}^{0}\left(s\right) =\varphi\left(x,y\right)
\end{align}
\end{linenomath}
for this known function, $\varphi$.
\begin{proposition}\label{prop:mainProposition}
If $\psi :S_{X}\rightarrow\mathbb{R}$ is a bounded, measurable function, then
\begin{linenomath}
\begin{align}\label{eq:mainPropEquation}
\sum_{t=0}^{\infty}\lambda^{t}\int_{\left(x,y\right)\in S_{X}\times S_{Y}} \left[ \psi\left(x\right) - \lambda\int_{s\in S_{X}}\psi\left(s\right)\,d\sigma_{X}\left[x,y\right]\left(s\right) \right] \, d\nu_{t}\left(x,y\right) &= \int_{s\in S_{X}}\psi\left(s\right) \, d\sigma_{X}^{0}\left(s\right) ,
\end{align}
\end{linenomath}
for any memory-one strategy, $\sigma_{X}$, where $\sigma_{X}^{0}$ is the initial action of player $X$.
\end{proposition}
\begin{proof}
Since $\psi$ is bounded, there exists a sequence of simple functions, $\left\{\psi_{n}\right\}_{n\geqslant 1}$, such that $\psi_{n}\rightarrow\psi$ uniformly on $S$. For each $n\geqslant 1$, let $\psi_{n}=\sum_{i=1}^{N_{n}}c_{i}^{n}\chi_{E_{i}^{n}}$. Using the uniform convergence of this sequence, together with the dominated convergence theorem and Lemma \ref{lem:akinsLemma}, we obtain
\begin{linenomath}
\begin{align}
\sum_{t=0}^{\infty} &\lambda^{t}\int_{\left(x,y\right)\in S_{X}\times S_{Y}} \left[ \psi\left(x\right) - \lambda\int_{s\in S_{X}}\psi\left(s\right)\,d\sigma_{X}\left[x,y\right]\left(s\right) \right] \, d\nu_{t}\left(x,y\right) \nonumber \\
&= \sum_{t=0}^{\infty}\lambda^{t} \lim_{n\rightarrow\infty} \int_{\left(x,y\right)\in S_{X}\times S_{Y}} \left[ \psi_{n}\left(x\right) - \lambda\int_{s\in S_{X}}\psi_{n}\left(s\right)\,d\sigma_{X}\left[x,y\right]\left(s\right) \right] \, d\nu_{t}\left(x,y\right) \nonumber \\
&= \lim_{n\rightarrow\infty} \sum_{t=0}^{\infty}\lambda^{t} \int_{\left(x,y\right)\in S_{X}\times S_{Y}} \left[ \psi_{n}\left(x\right) - \lambda\int_{s\in S_{X}}\psi_{n}\left(s\right)\,d\sigma_{X}\left[x,y\right]\left(s\right) \right] \, d\nu_{t}\left(x,y\right) \nonumber \\
&= \lim_{n\rightarrow\infty} \sum_{t=0}^{\infty}\lambda^{t}\sum_{i=1}^{N_{n}}c_{i}^{n}\int_{\left(x,y\right)\in S_{X}\times S_{Y}} \left[ \chi_{E_{i}^{n}}\left(x\right) - \lambda\int_{s\in S_{X}}\chi_{E_{i}^{n}}\left(s\right)\,d\sigma_{X}\left[x,y\right]\left(s\right) \right] \, d\nu_{t}\left(x,y\right) \nonumber \\
&= \lim_{n\rightarrow\infty} \sum_{t=0}^{\infty}\lambda^{t}\sum_{i=1}^{N_{n}}c_{i}^{n}\int_{\left(x,y\right)\in S_{X}\times S_{Y}} \Big[ \chi_{E_{i}^{n}}\left(x\right) - \lambda\sigma_{X}\left[x,y\right]\left(E_{i}^{n}\right) \Big] \, d\nu_{t}\left(x,y\right) \nonumber \\
&= \lim_{n\rightarrow\infty} \sum_{i=1}^{N_{n}}c_{i}^{n}\sum_{t=0}^{\infty}\lambda^{t}\int_{\left(x,y\right)\in S_{X}\times S_{Y}} \Big[ \chi_{E_{i}^{n}}\left(x\right) - \lambda\sigma_{X}\left[x,y\right]\left(E_{i}^{n}\right) \Big] \, d\nu_{t}\left(x,y\right) \nonumber \\
&= \lim_{n\rightarrow\infty} \sum_{i=1}^{N_{n}}c_{i}^{n}\sigma_{X}^{0}\left(E_{i}^{n}\right) \nonumber \\
&= \lim_{n\rightarrow\infty} \int_{s\in S_{X}}\psi_{n}\left(s\right)\,d\sigma_{X}^{0}\left(s\right) \nonumber \\
&= \int_{s\in S_{X}}\psi\left(s\right)\,d\sigma_{X}^{0}\left(s\right) ,
\end{align}
\end{linenomath}
which completes the proof.
\end{proof}

While Proposition \ref{prop:mainProposition} applies to discounted games with $\lambda <1$, we can get an analogous statement for undiscounted games by multiplying both sides of Eq. (\ref{eq:mainPropEquation}) by $1-\lambda$ and taking the limit $\lambda\rightarrow 1^{-}$:
\begin{corollary}
If $\psi:S_{X}\rightarrow\mathbb{R}$ is a bounded, measurable function, then, when the limit exists,
\begin{linenomath}
\begin{align}
\lim_{T\rightarrow\infty}\frac{1}{T+1}\sum_{t=0}^{T}\int_{\left(x,y\right)\in S_{X}\times S_{Y}} \left[ \psi\left(x\right) - \int_{s\in S_{X}}\psi\left(s\right)\,d\sigma_{X}\left[x,y\right]\left(s\right) \right] \, d\nu_{t}\left(x,y\right) &= 0
\end{align}
\end{linenomath}
for any memory-one strategy, $\sigma_{X}$, where $\sigma_{X}^{0}$ is the initial action of player $X$.
\end{corollary}

\begin{theorem}[Autocratic strategies in arbitrary action spaces]
Suppose that $\sigma_{X}\left[x,y\right]$ is a memory-one strategy for player $X$ and let $\sigma_{X}^{0}$ be player $X$'s initial action. If, for some bounded function, $\psi$, the equation
\begin{linenomath}
\begin{align}\label{eq:mainSIEquation}
\alpha u_{X}\left(x,y\right) + \beta u_{Y}\left(x,y\right) + \gamma &= \psi\left(x\right) - \lambda\int_{s\in S_{X}}\psi\left(s\right)\,d\sigma_{X}\left[x,y\right]\left(s\right) - \left(1-\lambda\right)\int_{s\in S_{X}}\psi\left(s\right)\,d\sigma_{X}^{0}\left(s\right)
\end{align}
\end{linenomath}
holds for each $x\in S_{X}$ and $y\in S_{Y}$, then $\sigma_{X}^{0}$ and $\sigma_{X}\left[x,y\right]$ together enforce the linear payoff relationship
\begin{linenomath}
\begin{align}
\alpha\pi_{X}+\beta\pi_{Y}+\gamma &= 0
\end{align}
\end{linenomath}
for \emph{any} strategy of player $Y$. In other words, the pair $\Big(\sigma_{X}^{0},\sigma_{X}\left[x,y\right]\Big)$ is an autocratic strategy.
\end{theorem}

\begin{proof}
If Eq. (\ref{eq:mainSIEquation}) holds, then by Eq. (\ref{eq:mainPropEquation}) in Proposition \ref{prop:mainProposition},
\begin{linenomath}
\begin{align}
\alpha\pi_{X}+&\beta\pi_{Y}+\gamma +\left(1-\lambda\right)\int_{s\in S_{X}}\psi\left(s\right)\,\sigma_{X}^{0}\left(s\right) \nonumber \\
&= \left(1-\lambda\right)\sum_{t=0}^{\infty}\lambda^{t}\int_{\left(x,y\right)\in S_{X}\times S_{Y}} \left[ \psi\left(x\right) - \lambda\int_{s\in S_{X}}\psi\left(s\right)\,d\sigma_{X}\left[x,y\right]\left(s\right) \right] \, d\nu_{t}\left(x,y\right) \nonumber \\
&= \left(1-\lambda\right)\int_{s\in S_{X}}\psi\left(s\right) \, d\sigma_{X}^{0}\left(s\right) ,
\end{align}
\end{linenomath}
and it follows at once that $\alpha\pi_{X}+\beta\pi_{Y}+\gamma =0$.
\end{proof}

\subsection{Two-point autocratic strategies}

\begin{corollary}\label{cor:mainCorollary}
Let $\alpha, \beta ,\gamma\in\mathbb{R}$ and suppose that there exist $s_{1},s_{2}\in S_{X}$ (i.e. two discrete actions) and $\phi >0$ such that
\begin{linenomath}
\begin{subequations}\label{eq:corollaryEquation}
\begin{align}
-\frac{1-\left(1-\lambda\right) p_{0}}{\phi} &\leqslant \alpha u_{X}\left(s_{1},y\right) + \beta u_{Y}\left(s_{1},y\right) + \gamma \leqslant -\frac{\left(1-\lambda\right)\left(1-p_{0}\right)}{\phi} ; \\
\frac{\left(1-\lambda\right) p_{0}}{\phi} &\leqslant \alpha u_{X}\left(s_{2},y\right) + \beta u_{Y}\left(s_{2},y\right) + \gamma \leqslant \frac{\lambda +\left(1-\lambda\right) p_{0}}{\phi}
\end{align}
\end{subequations}
\end{linenomath}
for each $y\in S_{Y}$, where $p_{0}$ is the probability that $X$ initially uses $s_{1}$ and $1-p_{0}$ is the probability that $X$ initially uses $s_{2}$. Let $\delta_{s}$ be the Dirac measure on $S$ centered at $s$, and, for $x\in\left\{s_{1},s_{2}\right\}$ and $y\in S_{Y}$, consider the memory-one strategy,
\begin{linenomath}
\begin{align}
\sigma_{X}\left[x,y\right] &:= p\left(x,y\right)\delta_{s_{1}} + \Big(1-p\left(x,y\right)\Big)\delta_{s_{2}} ,
\end{align}
\end{linenomath}
where
\begin{linenomath}
\begin{subequations}
\begin{align}
p\left(s_{1},y\right) &:= \frac{1}{\lambda}\Big(\phi\left(\alpha u_{X}\left(s_{1},y\right) +\beta u_{Y}\left(s_{1},y\right) +\gamma\right) -\left(1-\lambda\right) p_{0}+1\Big) ; \\
p\left(s_{2},y\right) &:= \frac{1}{\lambda}\Big(\phi\left(\alpha u_{X}\left(s_{2},y\right) +\beta u_{Y}\left(s_{2},y\right) +\gamma\right) -\left(1-\lambda\right) p_{0}\Big) .
\end{align}
\end{subequations}
\end{linenomath}
Then, irrespective of player $Y$'s strategy, this choice of $p_{0}$ and $\sigma_{X}\left[x,y\right]$ enforces $\alpha\pi_{X}+\beta\pi_{Y}+\gamma =0$.
\end{corollary}
\begin{proof}
By Theorem \ref{thm:mainTheorem}, we need only show that there exists $\psi :\left\{s_{1},s_{2}\right\}\rightarrow\mathbb{R}$ such that
\begin{linenomath}
\begin{align}
\alpha u_{X}\left(x,y\right) + \beta u_{Y}\left(x,y\right) + \gamma &= \psi\left(x\right) - \lambda\int_{s\in S_{X}}\psi\left(s\right)\,d\sigma_{X}\left[x,y\right]\left(s\right) -\left(1-\lambda\right)\int_{s\in S_{X}}\psi\left(s\right)\,d\sigma_{X}^{0}\left(s\right)
\end{align}
\end{linenomath}
for each $x\in S_{X}$ and $y\in S_{Y}$ in order to establish the equation $\alpha\pi_{X}+\beta\pi_{Y}+\gamma =0$. Indeed, since we are restricting $X$'s actions to two points, we may assume that $S_{X}=\left\{s_{1},s_{2}\right\}$. Fix $\psi\left(s_{1}\right) :=\psi_{1}\in\mathbb{R}$ and let $\psi\left(s_{2}\right) :=\frac{1}{\phi} +\psi_{1}$. Since
\begin{linenomath}
\begin{align}
\psi\left(s_{1}\right) - &\lambda\int_{s\in S_{X}}\psi\left(s\right)\,d\sigma_{X}\left[s_{1},y\right]\left(s\right) - \left(1-\lambda\right)\int_{s\in S_{X}}\psi\left(s\right)\,d\sigma_{X}^{0}\left(s\right) \nonumber \\
&= \psi_{1} - \lambda\left(\psi_{1}+\left(1-p\left(s_{1},y\right)\right)\frac{1}{\phi}\right) - \left(1-\lambda\right)\left(\psi_{1}+\left(1-p_{0}\right)\frac{1}{\phi}\right) \nonumber \\
&= \alpha u_{X}\left(s_{1},y\right) + \beta u_{Y}\left(s_{1},y\right) + \gamma
\end{align}
\end{linenomath}
and
\begin{linenomath}
\begin{align}
\psi\left(s_{2}\right) - &\lambda\int_{s\in S_{X}}\psi\left(s\right)\,d\sigma_{X}\left[s_{2},y\right]\left(s\right) - \left(1-\lambda\right)\int_{s\in S_{X}}\psi\left(s\right)\,d\sigma_{X}^{0}\left(s\right) \nonumber \\
&= \frac{1}{\phi} + \psi_{1} - \lambda\left(\psi_{1}+\left(1-p\left(s_{2},y\right)\right)\frac{1}{\phi}\right) - \left(1-\lambda\right)\left(\psi_{1}+\left(1-p_{0}\right)\frac{1}{\phi}\right) \nonumber \\
&= \alpha u_{X}\left(s_{2},y\right) + \beta u_{Y}\left(s_{2},y\right) + \gamma ,
\end{align}
\end{linenomath}
and since $0\leqslant p\left(x,y\right)\leqslant 1$ for each $x\in\left\{s_{1},s_{2}\right\}$ and $y\in S_{Y}$ by Eq. (\ref{eq:corollaryEquation}), the proof is complete.
\end{proof}

\begin{remark}
In the undiscounted case ($\lambda =1$), Eq. (\ref{eq:corollaryEquation}) is satisfied for some $\phi >0$ if and only if there exist $s_{1},s_{2}\in S_{X}$ such that
\begin{linenomath}
\begin{align}\label{eq:underOverZero}
\alpha u_{X}\left(s_{1},y\right) +\beta u_{Y}\left(s_{1},y\right) +\gamma &\leqslant 0 \leqslant \alpha u_{X}\left(s_{2},y\right) +\beta u_{Y}\left(s_{2},y\right) +\gamma
\end{align}
\end{linenomath}
for every $y\in S_{Y}$. Moreover, if Eq. (\ref{eq:corollaryEquation}) holds for some $\lambda$, $p_{0}$, $\phi$, and $s_{1},s_{2}\in S_{X}$, then it must be true that Eq. (\ref{eq:underOverZero}) also holds for every $y\in S_{Y}$. If Eq. (\ref{eq:underOverZero}) does not hold for a particular choice of $s_{1},s_{2}\in S_{X}$, then $s_{1}$ and $s_{2}$ cannot form a two-point autocratic strategy for any discounting factor, $\lambda$. Therefore, Eq. (\ref{eq:underOverZero}), which is easy to check, offers a straightforward way to show that two actions cannot form a two-point autocratic strategy for a particular game.
\end{remark}

\begin{remark}\label{rem:twoPointExplicit}
For $\psi$, $s_{1}$, and $s_{2}$, fixed, one can ask which strategies of the form
\begin{linenomath}
\begin{subequations}
\begin{align}
\sigma_{X}^{0} &= p_{0}\delta_{s_{1}} + \left(1-p_{0}\right)\delta_{s_{2}} ; \\
\sigma_{X}\left[x,y\right] &= p\left(x,y\right)\delta_{s_{1}}+\Big(1-p\left(x,y\right)\Big)\delta_{s_{2}} ,
\end{align}
\end{subequations}
\end{linenomath}
for some $p_{0}$ and $p\left(x,y\right)$, satisfy the equation
\begin{linenomath}
\begin{align}
\alpha u_{X}\left(x,y\right) + \beta u_{Y}\left(x,y\right) + \gamma &= \psi\left(x\right) - \lambda\int_{s\in S_{X}}\psi\left(s\right)\,d\sigma_{X}\left[x,y\right]\left(s\right) - \left(1-\lambda\right)\int_{s\in S_{X}}\psi\left(s\right)\,d\sigma_{X}^{0}\left(s\right) .
\end{align}
\end{linenomath}
Indeed, we see from the proof of Corollary \ref{cor:mainCorollary} that, for a strategy of this form, we must have
\begin{linenomath}
\begin{align}\label{eq:omegaExpression}
p\left(x,y\right) &= \frac{\psi\left(s_{2}\right) - \frac{1}{\lambda}\left(\psi\left(x\right) - \left(\alpha u_{X}\left(x,y\right) + \beta u_{Y}\left(x,y\right) + \gamma + \left(1-\lambda\right)\Big(\psi\left(s_{1}\right) p_{0} + \psi\left(s_{2}\right)\left(1-p_{0}\right)\Big)\right)\right)}{\psi\left(s_{2}\right) - \psi\left(s_{1}\right)}
\end{align}
\end{linenomath}
for each $x\in\left\{s_{1},s_{2}\right\}$ and $y\in S_{Y}$. Therefore, this simple case does not capture the generally-implicit nature of autocratic strategies because one can explicitly write down two-point strategies via Eq. (\ref{eq:omegaExpression}), which is typically not possible for strategies concentrated on more than just two actions.
\end{remark}

\section{Examples}\label{si:examples}

Here we present some simple examples of Theorem \ref{thm:mainTheorem} and its implications. In \S\ref{si:sizeFinite}, we demonstrate how Theorem \ref{thm:mainTheorem} reduces to the main result of \citet{press:PNAS:2012} when the action space has only two options. Moreover, we use an action space consisting of three choices to illustrate the implicit nature of autocratic strategies defined via Press-Dyson functions for more than two actions. In \S\ref{si:continuousDonationGame}, we show that there is no way for a player to unilaterally set her own score using Theorem \ref{thm:mainTheorem}. In particular, despite the implicit nature of autocratic strategies, one can use Theorem \ref{thm:mainTheorem} to deduce the non-existence of certain classes of strategies.

\subsection{Games with finitely many actions and no discounting}\label{si:sizeFinite}

Suppose that $S_{X}=S_{Y}=\left\{A_{1},\dots ,A_{n}\right\}$ is finite. If $\psi:S_{X}\rightarrow\mathbb{R}$ and $\sigma_{X}$ is a memory-one strategy, then, for $\lambda =1$,
\begin{linenomath}
\begin{align}
\psi\left(x\right) - &\int_{s\in S_{X}}\psi\left(s\right)\,d\sigma_{X}\left[x,y\right]\left(s\right) \nonumber \\
&= \psi\left(x\right) - \sum_{r=1}^{n}\psi\left(A_{r}\right)\sigma_{X}\left[x,y\right]\left(A_{r}\right) \nonumber \\
&= \psi\left(x\right) - \sum_{r=1}^{n-1}\psi\left(A_{r}\right)\sigma_{X}\left[x,y\right]\left(A_{r}\right) - \psi\left(A_{n}\right)\left(1-\sum_{r=1}^{n-1}\sigma_{X}\left[x,y\right]\left(A_{r}\right)\right) \nonumber \\
&= \psi\left(x\right) - \psi\left(A_{n}\right) - \sum_{r=1}^{n-1}\Big( \psi\left(A_{r}\right) -\psi\left(A_{n}\right) \Big)\sigma_{X}\left[x,y\right]\left(A_{r}\right) \nonumber \\
&= \begin{cases}\displaystyle\sum_{r=1}^{n-1}\Big( \psi\left(A_{n}\right) -\psi\left(A_{r}\right) \Big)\Big(\sigma_{X}\left[x,y\right]\left(A_{r}\right) -\delta_{r,r'}\Big) & x=A_{r'}\neq A_{n}, \\ \displaystyle\sum_{r=1}^{n-1}\Big( \psi\left(A_{n}\right) -\psi\left(A_{r}\right) \Big)\sigma_{X}\left[x,y\right]\left(A_{r}\right) & x=A_{n}.\end{cases}
\end{align}
\end{linenomath}
Therefore, if $n=2$ and $S_{X}=S_{Y}=\left\{A_{1},A_{2}\right\}$, we have
\begin{linenomath}
\begin{align}
\psi\left(x\right) - \int_{s\in S_{X}}\psi\left(s\right)\,d\sigma_{X}\left[x,y\right]\left(s\right) &= \Big(\psi\left(A_{2}\right) -\psi\left(A_{1}\right)\Big)\widetilde{\mathbf{p}} .
\end{align}
\end{linenomath}
Thus, a Press-Dyson function is a scalar multiple of $\widetilde{\mathbf{p}}$ (defined by Eq. (\ref{eq:pressDysonVector}) in the main text).

On the other hand, if $n=3$ and $S_{X}=S_{Y}=\left\{A_{1},A_{2},A_{3}\right\}$, then
\begin{linenomath}
\begin{align}
\psi\left(x\right) - \int_{s\in S_{X}}\psi\left(s\right)\,d\sigma_{X}\left[x,y\right]\left(s\right) &= \begin{cases}c_{1}\sigma_{X}\left[x,y\right]\left(A_{1}\right) + c_{2}\sigma_{X}\left[x,y\right]\left(A_{2}\right) - c_{1} & x=A_{1}, \\ c_{1}\sigma_{X}\left[x,y\right]\left(A_{1}\right) + c_{2}\sigma_{X}\left[x,y\right]\left(A_{2}\right) - c_{2} & x=A_{2}, \\ c_{1}\sigma_{X}\left[x,y\right]\left(A_{1}\right) + c_{2}\sigma_{X}\left[x,y\right]\left(A_{2}\right) & x=A_{3},\end{cases}
\end{align}
\end{linenomath}
where $c_{1}=\psi\left(A_{3}\right) -\psi\left(A_{1}\right)$ and $c_{2}=\psi\left(A_{3}\right) -\psi\left(A_{2}\right)$. For each $x\in S_{X}$ and $y\in S_{Y}$, the measure $\sigma_{X}\left[x,y\right]$ is uniquely determined by $\sigma_{X}\left[x,y\right]\left(A_{1}\right)$ and $\sigma_{X}\left[x,y\right]\left(A_{2}\right)$. Therefore, unlike for $n=2$, one cannot necessarily eliminate both of $c_{1}$ and $c_{2}$ by normalizing the Press-Dyson function. For this reason, as well as due to the fact that a Press-Dyson function reduces to a multiple of the Press-Dyson vector for $n=2$, it is perhaps more natural to refer to any multiple of $\widetilde{\mathbf{p}}$ (as opposed to just $\widetilde{\mathbf{p}}$ itself) as a Press-Dyson vector.

\subsection{Continuous Donation Game}\label{si:continuousDonationGame}

Here we establish some further results claimed in the main text for extortionate, generous, and equalizer strategies for the continuous Donation Game.

\subsubsection{Relationship to the classical Donation Game}

For a two-point strategy, $X$'s action space may be restricted to $S_{X}=\left\{0,K\right\}$. Therefore, the scaling function $\psi :S_{X}\rightarrow\mathbb{R}$ of Theorem \ref{thm:mainTheorem} is defined by two numbers, $\psi\left(0\right)$ and $\psi\left(K\right)$. Letting $\phi :=1/\left(\psi\left(0\right) -\psi\left(K\right)\right)$, we see by Corollary \ref{cor:mainCorollary} in \S\ref{subsec:detailedProofs} that the function
\begin{linenomath}
\begin{align}\label{eq:coopProb}
p\left(x,y\right) &:= \begin{cases}\frac{1}{\lambda}\left(1 - \phi\Big(\chi b\left(K\right) +c\left(K\right) - b\left(y\right) -\chi c\left(y\right) -\left(\chi -1\right)\kappa\Big) -\left(1-\lambda\right) p_{0}\right) & x=K, \\ \frac{1}{\lambda}\left(\phi\Big(b\left(y\right) +\chi c\left(y\right) + \left(\chi -1\right)\kappa\Big) -\left(1-\lambda\right) p_{0}\right) & x=0\end{cases}
\end{align}
\end{linenomath}
gives well-defined reaction probabilities provided $\psi$ (hence $\phi$) and $p_{0}$ are chosen so that $0\leqslant p\left(x,y\right)\leqslant 1$ for each $x\in\left\{0,K\right\}$ and $y\in\left[0,K\right]$. For any such $\psi$, the memory-one strategy,
\begin{linenomath}
\begin{align}\label{eq:twoPointDonation}
\sigma_{X}\left[x,y\right] &:= \Big( 1-p\left(x,y\right) \Big) \delta_{0} + p\left(x,y\right) \delta_{K} ,
\end{align}
\end{linenomath}
together with $\sigma_{X}^{0}=\left(1-p_{0}\right)\delta_{0}+p_{0}\delta_{K}$, defines an autocratic strategy that allows $X$ to enforce the linear relationship
\begin{linenomath}
\begin{align}\label{eq:extortionEquation}
\pi_{X}-\kappa &= \chi\left(\pi_{Y}-\kappa\right) .
\end{align}
\end{linenomath}
Note \eq{twoPointDonation} simply states formally that player $X$ fully cooperates with probability $p\left(x,y\right)$ and defects with probability $1-p\left(x,y\right)$ following an outcome in which $X$ plays $x$ and $Y$ plays $y$.

In the absence of discounting, i.e. in the limit $\lambda\rightarrow 1$, we have
\begin{linenomath}
\begin{subequations}
\label{eq:pgeneral}
\begin{align}
p\left(K,K\right) &= 1 - \phi\left(\chi -1\right)\Big( b\left(K\right) -c\left(K\right) -\kappa\Big) ; \label{eq:KKgeneral} \\
p\left(K,0\right) &= 1 - \phi\Big(\chi b\left(K\right) +c\left(K\right) -\left(\chi -1\right)\kappa\Big) ; \\
p\left(0,K\right) &= \phi\Big( b\left(K\right) +\chi c\left(K\right) +\left(\chi -1\right)\kappa\Big) ; \\
p\left(0,0\right) &= \phi\left(\chi -1\right)\kappa . \label{eq:00general}
\end{align}
\end{subequations}
\end{linenomath}
Thus, setting $b:=b\left(K\right)$, and $c:=c\left(K\right)$, the general form, \eq{pgeneral}, recovers the discrete-action-space case, \eq{pdiscrete}. In particular, the autocratic memory-one strategy in \eq{twoPointDonation} is a direct generalization of zero-determinant strategies to the continuous Donation Game. However, this strategy contains much more information than the corresponding strategy for the classical Donation Game because it encodes $X$'s play in response to $Y$'s for every $y\in\left[0,K\right]$. Despite the fact that player $Y$ has an uncountably infinite number of actions to choose from, player $X$ can still ensure that \eq{extortionEquation} holds by playing only two actions.

\subsubsection{Extortionate and generous strategies}

In the main text, we saw that $X$ can unilaterally enforce $\pi_{X}-\kappa =\chi\left(\pi_{Y}-\kappa\right)$ for $\chi\geqslant 1$ provided $0\leqslant\kappa\leqslant b\left(K\right) -c\left(K\right)$ and $\lambda$ is sufficiently close to $1$. If $\chi =1$, then $\kappa$ is irrelevant and the linear relationship is simply $\pi_{X}=\pi_{Y}$. Here, we show that, if $\chi >1$, then $0\leqslant\kappa\leqslant b\left(K\right) -c\left(K\right)$ is necessary for such a payoff relationship to be enforced via Theorem \ref{thm:mainTheorem}. Indeed, if
\begin{linenomath}
\begin{align}
\Big( u_{X}\left(x,y\right) -\kappa\Big) - \chi\Big( u_{Y}\left(x,y\right) -\kappa\Big) &= -\chi b\left(x\right) -c\left(x\right) + b\left(y\right) + \chi c\left(y\right) + \left(\chi -1\right)\kappa \nonumber \\
&= \psi\left(x\right) - \lambda\int_{s\in S_{X}}\psi\left(s\right)\,d\sigma_{X}\left[x,y\right]\left(s\right) -\left(1-\lambda\right)\int_{s\in S_{X}}\psi\left(s\right)\,d\sigma_{X}^{0}\left(s\right)
\end{align}
\end{linenomath}
for some bounded $\psi :S_{X}\rightarrow\mathbb{R}$ and each $x\in S_{X}$ and $y\in S_{Y}$, then we obtain
\begin{linenomath}
\begin{align}
\left(\chi -1\right) &\geqslant \sup\psi - \lambda\sup\psi - \left(1-\lambda\right)\sup\psi = 0
\end{align}
\end{linenomath}
by taking $y=0$ and letting $\psi\left(x\right)$ approach $\sup\psi$, and we see that
\begin{linenomath}
\begin{align}
-\chi b\left(K\right) -c\left(K\right) + b\left(K\right) +\chi c\left(K\right) + \left(\chi -1\right)\kappa &\leqslant \inf\psi - \lambda\inf\psi - \left(1-\lambda\right)\inf\psi = 0
\end{align}
\end{linenomath}
by taking $y=K$ and letting $\psi\left(x\right)$ approach $\inf\psi$. Since $\chi >1$, we have $0\leqslant\kappa\leqslant b\left(K\right) -c\left(K\right)$.

\begin{figure*}
\includegraphics[scale=0.35]{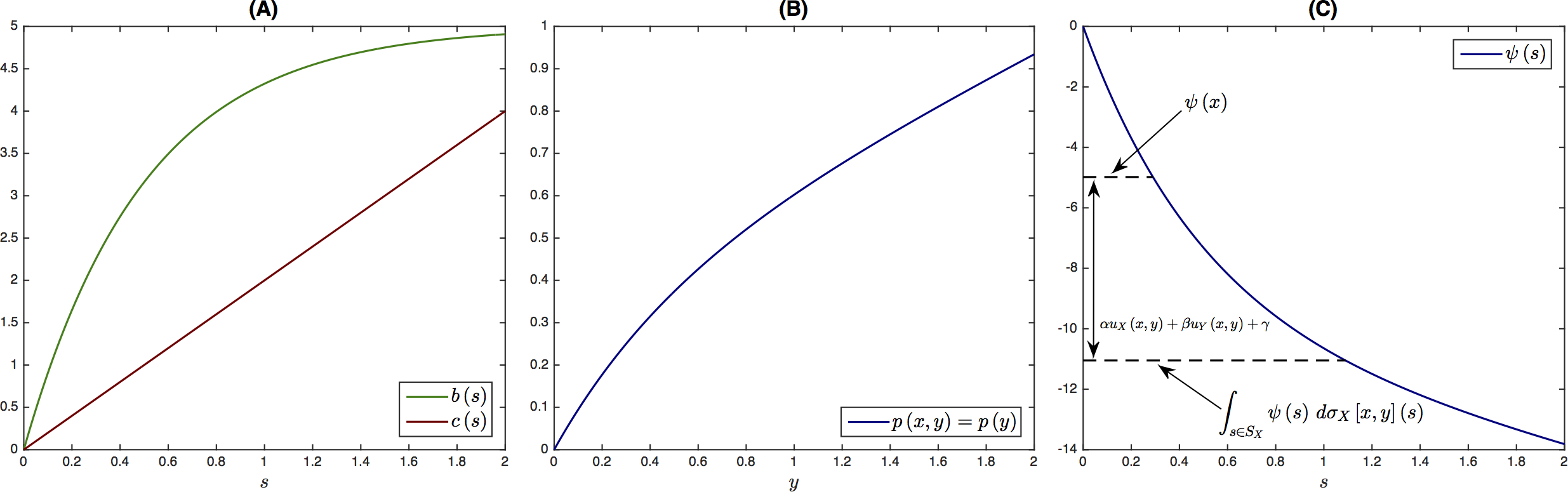}
\caption{Extortion in the continuous Donation Game without discounting. (A) shows linear costs, $c\left(s\right) =2s$, and saturating benefits, $b\left(c\left(s\right)\right) =5\left(1-e^{-c\left(s\right)}\right)$. In (B), the reaction function of player $X$, $p\left(x,y\right) =p\left(y\right)$ (see \eq{twoPointDonationExample}), indicates the probability that player $X$ plays $K$ (as opposed to $0$) after $Y$ uses $y$ and represents a reactive strategy because it depends solely on $Y$'s previous move. In (C), the function $\psi\left(s\right) =-\chi b\left(s\right) -c\left(s\right)$ is shown together with two dashed lines indicating $\psi(x)$ and $\int_{s\in S_{X}}\psi\left(s\right)\,d\sigma_{X}\left[x,y\right]\left(s\right)$, respectively, for $x=0.3$ where $\sigma_{X}\left[x,y\right]$ denotes the memory-one strategy for player $X$ based on the reactive strategy in (A). The vertical distance between the dashed lines must equal $\alpha u_{X}\left(x,y\right) +\beta u_{Y}\left(x,y\right) +\gamma =u_{X}\left(x,y\right) -\chi u_{Y}\left(x,y\right)$ for each $x,y\in\left[0,2\right]$ to satisfy \thm{mainTheorem}.
Parameters: $\lambda =1, \kappa =0, \alpha =1, \beta =-2$, and $\gamma =0$, yielding an extortion factor of $\chi :=-\beta =2$.\label{fig:donationGame}}
\end{figure*}

\subsubsection{Opponent-equalizing strategies}

Although a player cannot set her own score in the continuous Donation Game, she can set the score of her opponent. We saw in the main text that $X$ can set $Y$'s score to anything between $0$ and $b\left(K\right) -c\left(K\right)$ provided $\lambda$ is sufficiently large, and here we show that this interval is the only range of payoffs for player $Y$ that $X$ can unilaterally set via Theorem \ref{thm:mainTheorem}. Indeed, if $\gamma$ satisfies
\begin{linenomath}
\begin{align}
u_{Y}\left(x,y\right) - \gamma &= b\left(x\right) - c\left(y\right) - \gamma \nonumber \\
&= \psi\left(x\right) - \lambda\int_{s\in S_{X}}\psi\left(s\right)\,d\sigma_{X}\left[x,y\right]\left(s\right) -\left(1-\lambda\right)\int_{s\in S_{X}}\psi\left(s\right)\,d\sigma_{X}^{0}\left(s\right)
\end{align}
\end{linenomath}
for some bounded $\psi :S_{X}\rightarrow\mathbb{R}$ and each $x\in S_{X}$ and $y\in S_{Y}$, then we obtain
\begin{linenomath}
\begin{align}
-\gamma &\leqslant \inf\psi - \lambda\inf\psi - \left(1-\lambda\right)\inf\psi = 0
\end{align}
\end{linenomath}
by taking $y=0$ and letting $\psi\left(x\right)$ approach $\inf\psi$, and we find that
\begin{linenomath}
\begin{align}
b\left(K\right) - c\left(K\right) - \gamma &\geqslant \sup\psi - \lambda\sup\psi - \left(1-\lambda\right)\sup\psi = 0
\end{align}
\end{linenomath}
by taking $y=K$ and letting $\psi\left(x\right)$ approach $\sup\psi$. It follows immediately that $0\leqslant\gamma\leqslant b\left(K\right) -c\left(K\right)$.

\subsubsection{Self-equalizing strategies}

We saw in the main text that extortionate strategies exist in the continuous Donation Game, as demonstrated by the two-point strategy defined by Eq. (\ref{eq:twoPointDonationExample}). However, it certainly need not be the case that for any $\alpha ,\beta ,\gamma\in\mathbb{R}$, there exists $\psi:S_{X}\rightarrow\mathbb{R}$ such that
\begin{linenomath}
\begin{align}
\alpha u_{X}\left(x,y\right) + \beta u_{Y}\left(x,y\right) + \gamma &= \psi\left(x\right) -\lambda\int_{s\in S_{X}}\psi\left(s\right)\,d\sigma_{X}\left[x,y\right]\left(s\right) -\left(1-\lambda\right)\int_{s\in S_{X}}\psi\left(s\right)\,d\sigma_{X}^{0}\left(s\right)
\end{align}
\end{linenomath}
for some $\Big(\sigma_{X}^{0},\sigma_{X}\left[x,y\right]\Big)$ and each $x\in S_{X}$ and $y\in S_{Y}$. For example suppose that $\alpha =1$, $\beta =0$, and $\gamma$ is some fixed real number. For $u_{X}\left(x,y\right) =b\left(y\right) -c\left(x\right)$, the equation
\begin{linenomath}
\begin{align}
b\left(y\right) - c\left(x\right) - \gamma &= u_{X}\left(x,y\right) - \gamma = \psi\left(x\right) -\lambda\int_{s\in S_{X}}\psi\left(s\right)\,d\sigma_{X}\left[x,y\right]\left(s\right) - \left(1-\lambda\right)\int_{s\in S_{X}}\psi\left(s\right)\,d\sigma_{X}^{0}\left(s\right)
\end{align}
\end{linenomath}
holds for some $\psi:S_{X}\rightarrow\mathbb{R}$ and $\Big(\sigma_{X}^{0},\sigma_{X}\left[x,y\right]\Big)$, only if
\begin{linenomath}
\begin{align}
\lambda\inf\psi \leqslant \psi\left(x\right) + c\left(x\right) + \gamma - b\left(y\right) -\left(1-\lambda\right)\int_{s\in S_{X}}\psi\left(s\right)\,d\sigma_{X}^{0}\left(s\right) \leqslant \lambda\sup\psi
\end{align}
\end{linenomath}
for each $x,y\in\left[0,K\right]$. However, these inequalities imply that
\begin{linenomath}
\begin{align}
\left(1-\lambda\right)\inf\psi &\geqslant \left(1-\lambda\right)\sup\psi + \Big( b\left(K\right) - c\left(K\right) \Big) ,
\end{align}
\end{linenomath}
which is impossible since $b\left(K\right) >c\left(K\right)$ and $\inf\psi\leqslant\sup\psi$. Thus, there is no feasible Press-Dyson function that allows player $X$ to unilaterally enforce the equation $\pi_{X}=\gamma$. In other words, a player cannot unilaterally set her own payoff.

\subsubsection{Initial actions}

Here we state the conditions on $p_{0}$ (for two-point strategies) and $x_{0}$ (for deterministic strategies) that allow $X$ to enforce autocratic strategies in the continuous Donation Game.

If $p_{0}$ denotes the probability that $X$ initially plays $x=K$ (as opposed to $x=0$), then the function
\begin{linenomath}
\begin{align}
p(x,y) &= \frac{1}{\lambda}\left(\frac{b\left(y\right) + \chi c\left(y\right) +\left(\chi -1\right)\kappa -\left(1-\lambda\right)\Big(\chi b\left(K\right) +c\left(K\right)\Big) p_{0}}{\chi b\left(K\right) + c\left(K\right)}\right)
\end{align}
\end{linenomath}
gives a well-defined reaction probability provided $\lambda$ satisfies Eq. (\ref{eq:deltaConditionDonation}) in the main text and $p_{0}$ satisfies
\begin{linenomath}
\begin{align}\label{sieq:twoPointInitialBoundExtortion}
\max &\left\{\frac{\left(\chi -1\right)\kappa -\lambda\Big(\chi b\left(K\right) +c\left(K\right)\Big) + \Big(b\left(K\right) +\chi c\left(K\right)\Big)}{\left(1-\lambda\right)\Big(\chi b\left(K\right) +c\left(K\right)\Big)} , 0\right\} \nonumber \\
&\leqslant p_{0} \leqslant \min\left\{\frac{\left(\chi -1\right)\kappa}{\left(1-\lambda\right)\Big(\chi b\left(K\right) +c\left(K\right)\Big)} , 1\right\} .
\end{align}
\end{linenomath}
Moreover, for such a $p_{0}$ and $p\left(x,y\right)$, the two-point strategy defined by
\begin{linenomath}
\begin{subequations}
\begin{align}
\sigma_{X}^{0} &= \left(1-p_{0}\right)\delta_{0}+p_{0}\delta_{K} ; \\
\sigma_{X}\left[x,y\right] &= \Big(1-p\left(x,y\right)\Big)\delta_{0}+p\left(x,y\right)\delta_{K}
\end{align}
\end{subequations}
\end{linenomath}
allows $X$ to enforce $\pi_{X}-\kappa =\chi\left(\pi_{Y}-\kappa\right)$. This reaction probability was obtained using the scaling function $\psi\left(s\right) :=-\chi b\left(s\right) -c\left(s\right)$ on $S_{X}=\left\{0,K\right\}$.

If $X$ wishes to enforce $\pi_{Y}=\gamma$ for $0\leqslant\gamma\leqslant b\left(K\right) -c\left(K\right)$, she may use $\psi\left(s\right) :=b\left(s\right)$ on $S_{X}=\left\{0,K\right\}$ and
\begin{linenomath}
\begin{align}
p\left(x,y\right) &= \frac{1}{\lambda}\left(\frac{c\left(y\right) +\gamma}{b\left(K\right)}-\left(1-\lambda\right) p_{0}\right) ,
\end{align}
\end{linenomath}
provided $\lambda\geqslant c\left(K\right) /b\left(K\right)$ and
\begin{linenomath}
\begin{align}\label{sieq:twoPointInitialBoundEqualizer}
\max\left\{\frac{\gamma -\lambda b\left(K\right) +c\left(K\right)}{\left(1-\lambda\right) b\left(K\right)} , 0\right\} \leqslant p_{0} \leqslant \min\left\{\frac{\gamma}{\left(1-\lambda\right) b\left(K\right)} , 1\right\} .
\end{align}
\end{linenomath}

Using $\psi\left(s\right) :=-\chi b\left(s\right) -c\left(s\right)$ on $S_{X}=\left[0,K\right]$, $X$ can also enforce $\pi_{X}-\kappa =\chi\left(\pi_{Y}-\kappa\right)$ using the deterministic strategy with reaction function
\begin{linenomath}
\begin{align}
r^{X}\left(x,y\right) &= \psi^{-1}\left(\frac{-b\left(y\right) -\chi c\left(y\right) -\left(\chi -1\right)\kappa -\left(1-\lambda\right)\psi\left(x_{0}\right)}{\lambda}\right) ,
\end{align}
\end{linenomath}
where $\psi^{-1}\left(\cdots\right)$ denotes the inverse of the function $\psi$, provided Eq. (\ref{eq:deltaConditionDonation}) holds and $X$'s initial action, $x_{0}$, satisfies
\begin{linenomath}
\begin{align}\label{sieq:deterministicInitialBoundExtortion}
\frac{\left(\chi -1\right)\kappa - \lambda\Big(\chi b\left(K\right) +c\left(K\right)\Big) +\Big(b\left(K\right) +\chi c\left(K\right)\Big)}{1-\lambda} \leqslant \chi b\left(x_{0}\right) +c\left(x_{0}\right) \leqslant \frac{\left(\chi -1\right)\kappa}{1-\lambda} .
\end{align}
\end{linenomath}

For $X$ to set $\pi_{Y}=\gamma$ for $0\leqslant\gamma\leqslant b\left(K\right) -c\left(K\right)$ using a deterministic strategy, then she can do so by using $\psi\left(s\right) :=b\left(s\right)$ and the reaction function
\begin{linenomath}
\begin{align}
r^{X}\left(x,y\right) &= \psi^{-1}\left(\frac{c\left(y\right) +\gamma -\left(1-\lambda\right)\psi\left(x_{0}\right)}{\lambda}\right) ,
\end{align}
\end{linenomath}
provided $\lambda\geqslant c\left(K\right) /b\left(K\right)$ and $X$'s initial action, $x_{0}$, satisfies
\begin{linenomath}
\begin{align}\label{sieq:deterministicInitialBoundEqualizer}
\frac{\gamma -\lambda b\left(K\right) +c\left(K\right)}{1-\lambda} \leqslant b\left(x_{0}\right) \leqslant \frac{\gamma}{1-\lambda} .
\end{align}
\end{linenomath}

\end{document}